\documentclass[12pt]{article}
\usepackage{amsmath, amssymb}
\usepackage{graphicx}
\usepackage{natbib}
\usepackage{url} 
\usepackage{booktabs,multirow,multicol}

\usepackage{xcolor}

\renewcommand{\baselinestretch}{1.5}

\begin{document}

\def\spacingset#1{\renewcommand{\baselinestretch}%
{#1}\small\normalsize} \spacingset{1}

\newtheorem{lem}{Lemma}
\newtheorem{cor}{Corollary}
\newtheorem{prp}{Proposition}
\newtheorem{asm}{Assumption}
\newtheorem{thm}{Theorem}
\newtheorem{rem}{Remark}


  \title{\bf Outcome-dependent density combination by mixture Bayesian predictive synthesis}
  \author{Yuta Matano\\
    Graduate School of Economics, University of Tokyo \\ 
    \url{matano-yuta729@g.ecc.u-tokyo.ac.jp} \\
    and \\
    Kaoru Irie\\
    Faculty of Economics, University of Tokyo \\ 
    \url{irie@e.u-tokyo.ac.jp} \\
    }
    \date{}
  \maketitle

\bigskip
\begin{abstract}
In the problem of density combination, the decision maker receives multiple density functions and is tasked to learn which density is used most in which region of the outcome space, and to combine them into a single density. 
The outcome-dependent density combination can be cast as the Bayesian predictive synthesis, where the decision maker's uncertainty is naturally integrated. 
The models for the outcome-dependent weights have been considered in \cite{johnsonwest2025outcome}, but limited to the Gaussian density kernel. 
In this study, we point out the limitation of the model flexibility of the Gaussian weights, and theoretically and numerically verify that the estimated weights can be least outcome-dependent. As an alternative, we propose the use of the softmax (logistic-normal) functions as the models for the weights, together with the bias-adjustment to the submitted densities. The posterior computation is feasible by using the custom Metropolis-Hastings method with the P\'olya-gamma augmentation. The utility of the proposed method is illustrated in the real data analysis of movie theatre revenues, where a human expert reports the point quantile forecasts without any uncertainty quantification. 
\end{abstract}

\noindent%
{\it Keywords:} Density combination; outcome-dependent mixture model; Bayesian predictive synthesis; softmax transformation; P\'olya-gamma mixture; Quantile combination.
\vfill

\spacingset{1.5} 

\clearpage 
\section{Introduction}

In predictive analysis, it is common to have multiple predictive distributions of the quantity of interest from different sources of information, such as models, computer algorithms, and/or forecasts by human experts. 
For example, macroeconomic indices such as GDP can be forecast by multiple macroeconometric models (e.g., \citealt{chernis2022nowcasting}), in addition that the expectations of professionals on those indices are surveyed and integrated into the model-based predictions (e.g., \citealt{adrian2025scenario}). 
Since various models and experts provide different forecasts, it is reasonable to assume that one forecast describes the typical values, while another is informative about outliers. 
Considering the diversity of the provided predictive distributions, we, as the decision maker, are then tasked to combine these distributions into a single distribution in an outcome-dependent way, while learning which distribution is relevant in which region of the outcome space. 

Standard density-combination methods offer only a partial answer to this problem of outcome-dependent density combination. 
Linear and logarithmic pools, together with score-based prediction pools, are widely used to construct aggregate density forecasts \citep{genestzidek1986combining,hallmitchell2007combining,gewekeamisano2011optimal,gneitingranjan2013combining}. 
However, these methods consider the mixture of the densities using the constant, outcome-independent weights. 
To address the outcome-dependent weights, the generalized density forecast combination of \citet{kapetanios2015generalised} partitions the outcome space and assigns region-specific weights to the candidate densities. 
However, such pooling rules are not, by themselves, Bayesian updating: they construct an aggregate density directly rather than deriving a posterior predictive distribution conditional on their information set.

A natural Bayesian solution to the problem of density combination is the Bayesian model averaging, where the prior over the models, algorithms and experts is assumed in order to derive the combined density as the posterior. However, it is well known that the predictive density based on the model averaging is outcome-independent. 
In this context, Bayesian predictive synthesis (BPS) considers the prior over the densities/distributions, providing an outcome-dependent but coherent way of combining the densities. 
Originating in Bayesian agent opinion analysis \citep{genest1985modeling,West1992d,westcrosse1992opinion}, BPS introduces a custom latent variable for each density and combines them through a synthesis function. 
The influential dynamic BPS model of \citet{mcalinn2019dynamic} develops this idea through a dynamic regression synthesis function. 
Subsequent work extends the BPS framework in several directions, including multivariate synthesis \citep{mcalinn2020multivariate}, sequential computation \citep{masudairie2023sequential}, and application-specific synthesis models \citep{cabel2022bayesian,sugasawa2023causal}. 
However, modeling of the synthesis function for the outcome-dependent density combination has not been explored extensively. 

In this study, we consider the BPS-based approach to the outcome-dependent density combination, including the models for weight functions, posterior computation and illustration by applications. 
One key development in the BPS-based density combination is the outcome-dependent generalized opinion pools of \citet{johnsonwest2025outcome}, covering issues related to uncertainty quantification such as model-set incompleteness. 
The outcome-dependent weights are modeled by using the Gaussian density kernel, so that each candidate density is used most in some neighborhood, but less outside. 
Although the Gaussian weights were reported to provide reasonable predictions in time series analysis, their direct application to density combination was not studied. In fact, we have confirmed in simulation studies that the estimates of the Gaussian weights can be constant, hence outcome-independent. We mathematically verify this phenomenon using the theory of posterior concentration in a large sample situation. 

We then consider an alternative BPS-based model for density combination, beyond the Gaussian models. First, we enable the synthesis function to account for our uncertainty on the reported forecast densities, as considered in \cite{johnson2017phdthesis}. This modification not only improves the model flexibility but also simplifies the posterior computation by the Markov chain Monte Carlo method. Second, we propose to use softmax functions, or logistic-normal density functions, as the model for outcome-dependent weights. Unlike the Gaussian weights, the softmax weights are unbounded, allowing the sharp switch from one candidate density to another. The posterior computation is straightforward using the P\'olya-gamma augmentation technique \citep{polson2013bayesian}. Taken together, the flexibility of the proposed model can handle a challenging problem of density combination, such as the combination of point forecasts about density quantiles, as illustrated later using the dataset of movie theatre revenue. 

The remainder of the paper is organized as follows. Section~2 reviews density forecast combinations, BPS, and outcome-dependent mixture synthesis. Section~3 studies Gaussian-kernel outcome-dependent weights and their undesirable behavior in a large sample situation. Section~4 defines the softmax weight models and discusses their difference from the Gaussian weights. Section~5 develops custom Markov chain Monte Carlo algorithms for posterior computation. Section~6 presents the results of numerical experiments for comparative studies. Section~7 is an illustration of the proposed method in a real data analysis. Section~8 concludes with possible future research.

\section{Bayesian predictive synthesis}

In this section, we briefly review the framework of Bayesian predictive synthesis (BPS). The density combination is viewed as the BPS of a certain type and called mixture BPS. Specifically, we consider two models of the mixture BPS: Dirac and bias-adjusted synthesis functions. 

\subsection{Density combination as BPS}

Suppose that we are interested in the prediction of univariate quantity $y$ and, for this predictive purpose, given $K$ probability densities, $h_1(\cdot ), \dots , h_K(\cdot)$. These densities define the information set, denoted by $\mathcal{H}=\{ h_1,\dots, h_K \}$. 
As noted in the introduction, these densities can be provided from different sources, including statistical/econometric models, deterministic algorithms and/or human experts, which are all called {\it agents}.  
Since $\mathcal{H}$ is known, the uncertainty about $y$ should be quantified by its conditional distribution, or ``posterior'', as $p(y|\mathcal{H})$. In modeling this distribution, however, one must specify the joint distribution of $y$ and $\mathcal{H}$, which is often a difficult task. Instead, a series of research on predictive synthesis by \cite{genest1985modeling} and \cite{westcrosse1992opinion} considered a situation where such distribution is specified only {\it partially}. It has been proven that, for such a model to be coherent in uncertainty quantification, the posterior distribution of interest must be of the form, 
\begin{equation*} 
    p(y|\mathcal{H}) = \int \alpha (y|x) h(x) dx,
\end{equation*}
where $x=(x_1,\dots,x_K)^{\top}$, 
$\alpha(y\mid x)$ is a conditional density called the synthesis function, and 
\begin{equation*}
     h(x) = \prod _{k=1}^K h_k(x_k).
\end{equation*}
This representation implies that modeling of $(y,\mathcal{H})$ is equivalent to specifying $\alpha (y|x)$. \cite{mcalinn2017dynamic}, \cite{mcalinn2019dynamic} and other recent research on BPS have studied various models for $\alpha (y|x)$ and the functional forms of the implied predictive density $p(y|\mathcal{H})$.

\subsection{Modeling of synthesis functions} \label{sec:syngen}

By choosing synthesis function $\alpha(y|x)$ appropriately, typical models for density combination can be viewed as BPS, hence justified as a coherent way of quantifying predictive uncertainty. 
In modeling the synthesis function, we specify a class of functions indexed by unknown quantities, $\theta$, and parametrize the synthesis function as $\alpha (y|x,\theta)$. Since $\theta$ is unknown to us, we quantify our uncertainty on $\theta$ via our posterior distribution using some test data. In this process of posterior computation, we evaluate the predictive performance of each agent in $\mathcal{H}$. 

The predictive analysis in the BPS framework is formulated as follows. The BPS equation is rewritten with explicit dependence on parameter $\theta$ as 
\begin{equation*} 
    p(y| \mathcal{H} , \theta ) = \int \alpha (y|x,\theta ) h(x) dx.
\end{equation*}
Then, test data of size $n$, or $y_1,\dots ,y_n$, are assumed to be generated independently from this distribution. With prior $\Pi(\theta)$, the posterior distribution of $\theta$ is obtained as 
\begin{equation*}
    p(\theta | y_{1:n} , \mathcal{H} ) \propto \Pi(\theta ) \prod _{i=1}^n \int \alpha (y_i|x_i, \theta ) h(x_i) dx_i,
\end{equation*}
where $x_i = (x_{i1},\dots ,x_{iK})^{\top}$ for $i=1,\dots ,n$. This density is often numerically approximated using a Markov chain Monte Carlo (MCMC) method with sampling of latent variables $x_{1:n}$ and parameter $\theta$. Finally, the predictive distribution of our interest is expressed as 
\begin{equation*}
    p(y|y_{1:n},\mathcal{H}) = \int\!\!\!\int \alpha (y|x,\theta ) h(x) p(\theta | y_{1:n} , \mathcal{H} ) dxd\theta ,
\end{equation*}
which is also evaluated by simulation. 

For synthesis function $\alpha (y|x,\theta )$, we consider two models that correspond to the outcome-dependent density combination. 

\paragraph{Dirac synthesis:} Following \cite{johnsonwest2025outcome}, we first consider the synthesis function of the form,
\begin{equation*}
    \alpha (y|x,\theta) = w_0(x;\theta ) h_0(y) + \sum _{k=1}^K w_k(x;\theta ) \delta _{x_k}(y),
\end{equation*}
where $\delta_{x_k}(y)$ is the degenerate distribution of $y$ at $x_k$, and $w_0(x ; \theta )$, $w_1(x ; \theta )$, $\dots$, $w_K(x ; \theta )$ are the weight functions; they are non-negative and sum to unity for any $x$ and $\theta$. 
Here, we have added an extra density, $h_0(y)$, to the mixture form. This density is called a baseline density, represents our own prior belief on $y$ and complements the incompleteness of $\mathcal{H}$ as the model set. 
The functional form of $w_k(\cdot ; \theta )$ will be given later in Sections~\ref{sec:gaussian} and \ref{sec:softmax}. We call this model {\it Dirac synthesis}. 

Note that, in general, the agent $k$'s weight function can depend on not only $x_k$ but also the other $x_{k'}$ ($k'\not=k$). If one assumes the independence of the agents in the weights and restricts the weight function as $w_k(x_k;\theta)$, the integral in $p(y|\mathcal{H},\theta )$ is simplified as 
\begin{equation} \label{eq:diracpred}
    p(y|\mathcal{H},\theta ) = w_0(\theta ) h_0(y) + \sum _{k=1}^K w_k(y;\theta ) h_k(y),
\end{equation}
where $w_0(\theta )$ is the integral of $w_0(x;\theta )$ with respect to $h(x)$. This is exactly the outcome-dependent density combination of $\mathcal{H}=\{ h_1,\dots, h_K \}$ (and baseline $h_0$) as desired. Note also that this form can be viewed as the outcome-{\it independent} combination of tilted agent forecasts, whose density kernels are given by $w_k(y;\theta) h_k(y)$. 

\paragraph{Bias-adjusted synthesis:} 
The synthesis function can be parametrized further to increase its predictive flexibility by replacing $\delta_{x_k}(y)$ in the Dirac synthesis with some continuous distribution as  
\begin{equation} \label{eq:biasad}
    \alpha (y|x,\theta) = w_0(x;\theta ) h_0(y) + \sum _{k=1}^K w_k(x;\theta ) \mathcal{N}(y|x_k-\beta_k,\tau_k^2),
\end{equation}
where $\mathcal{N}(y|\mu, \sigma^2)$ is the density of the normal distribution with mean $\mu$ and variance $\sigma^2$ evaluated at $y$. Compared to the Dirac synthesis, this functional form reflects our intention not to take the reported prediction $x_k$ from agent $h_k$ as-is, but to account for potential bias by $\beta_k$ and extra uncertainty by additive noise with variance $\tau^2_k$. The synthesis with location bias was considered in \cite{johnsonwest2025outcome} and extra uncertainty in \cite{johnson2017phdthesis} in the case of the Gaussian weight specification discussed later. We name this formulation as {\it Bias-adjusted synthesis}. 

The utility of this formulation is evident when the forecasts are reported as point predictions. Suppose that the $k$-th agent reports its forecast not in the form of a probability density, but a point forecast value, $\hat{y}_k$, with no information on its uncertainty. In the BPS framework, this scenario translates as $h_k(\cdot ) = \delta _{\hat{y}_k}(\cdot )$. In the Dirac synthesis, the combined density in (\ref{eq:diracpred}) is simply a mixture of those point masses (although it could be outcome-dependent). However, in the bias-adjusted synthesis, we have 
\begin{equation}\label{eq:deg_ba}
    p(y|\mathcal{H},\theta,\beta ,\tau ^2) = w_0(\hat{y}; \theta ) h_0(y) + \sum _{k=1}^K w_k(\hat{y};\theta ) \mathcal{N}(y|\hat{y}_k-\beta_k,\tau_k^2),
\end{equation}
where $\hat{y} = (\hat{y}_1, \ldots , \hat{y}_K)^\top$. In this expression, the point mass distribution is smoothed, similarly to the density estimation with the Gaussian kernels. 

Although $(\beta,\tau^2)$ is supposed to address the bias in point forecast and uncertainty quantification of each agent predictive density, the role of those parameters in fitting data $y$ overlaps with that of the other parameter $\theta$ in the weight function. Put differently, $\theta$ and $(\beta ,\tau ^2)$ could be unidentifiable. This conflict of the parameters is clarified by marginalizing $x$ out and rewriting the model in the outcome-{\it independent} form: 
\begin{equation*}
    p(y|\mathcal{H},\theta ,\beta ,\tau ^2) = w_0(\theta ) h_0(y) + \sum _{k=1}^K \int w_k(x ;\theta ) \mathcal{N}(x_k|y+\beta_k,\tau_k^2) h(x) dx. 
\end{equation*}
If the weight functions are modeled carefully, both Dirac and bias-adjusted synthesis models should fit to the data equivalently. 
Nevertheless, we consider the bias-adjusted synthesis for computational reasons; as pointed out later in Section~\ref{sec:mcmc}, the posterior computation under the Dirac synthesis by the MCMC method is complicated due to the degenerate distributions. 
To mitigate the problem of the identifiability, we use an informative (not diffuse) prior for $(\beta ,\tau ^2)$ and maintain the interpretability of the parameter $\theta$ and weight functions.

\section{Gaussian-kernel outcome-dependent weights}
\label{sec:gaussian}

In this section, we review the Gaussian-kernel specification for outcome-dependent BPS weights proposed by \citet{johnsonwest2025outcome}, then point out its potentially-undesirable behavior in a large-sample situation. 

\subsection{Gaussian weight functions}
\label{sec:gaussian-spec}

For each agent $k=1,\dots,K$, we specify the weight function by the Gaussian kernel as 
\begin{equation}
    w_k(x_k;\theta)
    = q_k \exp\left\{ -\frac{(x_k-\mu_k)^2}{2\sigma_k^2} \right\},
    \label{eq:gauss-weight}
\end{equation}
where $\theta=(q_{1:K},\mu_{1:K},\sigma_{1:K}^2)$, $\mu_k \in \mathbb{R}$, $\sigma_k>0$ and $q_{1:k}$ are the weights ($q_{1:k}$ are non-negative and sum to unity). The baseline weight is defined as the residual: $w_0(x;\theta)=1-\sum_{k=1}^K w_k(x_k;\theta)$.
Since $0\le w_k(x_k;\theta)\le q_k$ for every $x_k$, we have $w_0(x;\theta)\in[0,1]$ for every $x$, and $w_0,w_1,\dots ,w_K$ are on the simplex, assuring that $p(y|\mathcal{H},\theta)$ (or $p(y|\mathcal{H},\theta,\beta,\tau ^2)$) is a proper conditional density. Parameter $\mu_k$ is the point at which agent $k$ receives its maximal weight, $\sigma_k^2$ is a ``window'' that controls how quickly that weight decays away from $\mu_k$, and $q_k$ is the largest pointwise weight that agent $k$ can receive.

In analyzing the posterior of $\theta$ and/or $(\beta,\tau ^2)$, it is necessary to obtain the outcome-independent form by marginalizing latent $x$ out (i.e., the expression in (\ref{eq:diracpred})). The closed form is computed by algebra; see Section~\ref{supp:marginal}.

\subsection{Properties of Gaussian weights}
\label{sec:gaussian-weight}

The flexibility of the Gaussian specification (\ref{eq:gauss-weight}) is limited due to the weight cap $q_k$ and relation to baseline $h_0(\cdot)$ and window $\sigma_k$. 

\paragraph{Weight cap $q_k$.}
The weight function is inherently upper-bounded globally as 
\begin{equation*}
    0 \le w_k(x_k;\theta) \le q_k,
    \qquad
    \mathrm{for \ all \ } x_k \in \mathbb{R}. 
\end{equation*}
Hence $w_k(\cdot ;\theta)$ never exceeds $q_k$. In addition, if $q_k$ increases, then $q_{k'}$ must decrease for some $k'\not=k$, down-weighting the importance of $k'$ globally. 

To see this point, consider an example of $K=2$, where the data generating process uses the switching rule as the weights: 
$w_1^{\mathrm{sw}}(y)=\mathbb{I}(y<0)$ and
$w_2^{\mathrm{sw}}(y)=\mathbb{I}(y\ge0)$. 
In this example, agent 1 is used exclusively on the left side of the real line and agent 2 on the right. 
The Gaussian weights in (\ref{eq:gauss-weight}) cannot approximate this scenario satisfactorily, for the weight of agent 1 is at most $q_1$ and that of agent 2 is $q_2=1-q_1$ everywhere, no matter what values of $(\mu _{1:2},\sigma^2 _{1:2})$ are chosen. As seen in this example, the family of the Gaussian weights in fact excludes outcome-dependent ways of selecting agents by construction. Although not discussed here in detail, the same problem can be seen in other variants of Gaussian weights, including the cross-model weights in \cite{johnsonwest2025outcome}. 

\paragraph{Baseline and window $\sigma_k$.}
Parameter $\sigma_k$ is supposed to be a window that allocates the weight of agent $k$ around $\mu_k$. By lowering the value of $\sigma_k$, one might expect that agent $k$ is intensively used around $\mu_k$, but this interpretation is not necessarily correct.  

To see this point, take a simple example of the Dirac synthesis where the density of agent $k$ is also Gaussian; $h_k(x_k) = \mathcal{N}(x_k|m_k,s_k^2)$. Then, consider the weight in the outcome-independent form by marginalizing $x$ out. If $m_k=\mu_k$, then the tilted agent density of $k$ and its outcome-independent weight are computed as 
\begin{equation}
    \int w_k(x_k;\theta) \delta_{x_k}(y) h(x)dx = q_k \sqrt{\frac{\sigma_k^2}{\sigma_k^2+s_k^2}} \mathcal{N}\!\left( y\left| m_k, \frac{s_k^2\sigma_k^2}{s_k^2 + \sigma_k^2} \right.\right).
    \label{eq:c1-example}
\end{equation}
Here, the weight of agent $k$ is $q_k\sqrt{\sigma_k^2/(\sigma_k^2+s_k^2)}$, and it converges to zero as $\sigma_k^2\downarrow0$. 
Note also that the weights of the other agents do not increase when $\sigma_k^2\downarrow0$. 
Thus, the contribution of agent $k$ vanishes, and its share is reallocated to the baseline $h_0$. 
In \citet{johnsonwest2025outcome}, $h_0$ is introduced as a fallback component against model-set incompleteness, and a diffuse baseline is a natural choice for that role. Under the Gaussian weight specification, however, an overly diffuse $h_0$ could result in a diffuse predictive density, ignoring the agent densities in $\mathcal{H}$.

\subsection{Asymptotic posterior of weight parameters}
\label{sec:gaussian-asymptotics}

As seen later in Section~\ref{sec:gaussian-illustration}, the posterior estimate of $\sigma_k^2$ tends to be large, which makes the weight functions constant hence not outcome-dependent at all. 
We theoretically verify this numerical phenomenon through the posterior-concentration for the weight parameters in the large-sample situation. This theoretical result applies to a general class of the synthesis functions, including the Dirac/bias-adjusted synthesis with the Gaussian weights discussed here.

Recall from Section~2 that a parameterized synthesis function induces the likelihood of the form,
\begin{equation*}
    m_\theta(y)
    := p(y\mid\mathcal{H},\theta)
    = \int \alpha(y\mid x,\theta) h(x)\,dx.
\end{equation*}
The test data are assumed to be $Y_1,\ldots,Y_n\overset{\mathrm{iid}}{\sim}P^\ast$ (with density $p^{\ast}$). In this case, the likelihood is 
\begin{equation*}
    L_n(\theta) = \prod_{i=1}^n m_\theta(Y_i).
\end{equation*}
With the prior $\Pi$ for $\theta$ defined on a metric space $(\Theta,d)$, the posterior is obtained as $\Pi_n(\cdot\mid Y_{1:n},\mathcal{H})$, where information set $\mathcal{H}$ is fixed. Define the objective function as 
\[
    M(\theta)
    =
    \mathbb{E}_{P^\ast}\log m_\theta(Y)
    =
    \mathbb{E}_{P^\ast}\log p^\ast(Y)
    -
    \mathrm{KL}(p^\ast,m_\theta),
\]
assuming that the relevant expectations are finite.
The second term of the right is the Kullback-Leibler (KL) divergence from $p^{\ast}$ to $m_\theta$. 
Thus maximizing $M$ is equivalent to projecting the true predictive law onto the BPS family in KL divergence. Let
\begin{equation*}
    M^\star = \sup_{\theta\in\Theta} M(\theta),
    \qquad {\rm where} \quad
    \Theta^\star = \{\theta\in\Theta : M(\theta)=M^\star\}.
\end{equation*}
Under misspecification, $\Theta^\star$ is a set of pseudo-true parameter values.

\begin{prp}[Posterior concentration for parameterized BPS weights]
\label{prp:bps-weight-concentration}
Suppose that the following conditions hold: (i) $m_\theta$ is a proper density, $m_\theta(Y)>0$ $P^\ast$-almost surely, and $M(\theta)$ is finite for every $\theta\in\Theta$; (ii) the log score satisfies the uniform law of large numbers
\begin{equation*}
    \sup_{\theta\in\Theta}
    \left|
    \frac{1}{n}\sum_{i=1}^n \log m_\theta(Y_i) - M(\theta)
    \right|
    \to 0
    \quad\text{in } P^\ast\text{-probability};
\end{equation*}
(iii) $M^\star$ is finite and $\Theta^\star$ is nonempty; (iv) $M$ is well separated at $\Theta^\star$, in the sense that for every $\varepsilon>0$ there exists $\delta_\varepsilon>0$ such that
\begin{equation*}
    \sup_{\theta:\,d(\theta,\Theta^\star)\ge\varepsilon} M(\theta)
    \le M^\star - \delta_\varepsilon;
\end{equation*}
and (v) the prior charges all near-optimal sets,
\begin{equation*}
    \Pi\{\theta\in\Theta: M(\theta)\ge M^\star-\eta\}>0,
    \qquad \eta>0.
\end{equation*}
Then, for every $\varepsilon>0$,
\begin{equation*}
    \Pi_n\{\theta: d(\theta,\Theta^\star)\ge\varepsilon
    \mid Y_{1:n},\mathcal{H}\}
    \to 0
    \quad\text{in }P^\ast\text{-probability}.
\end{equation*}
In particular, if $\Theta^\star=\{\theta^\star\}$, the posterior is consistent at the pseudo-true value $\theta^\star$.
\end{prp}

The proof and discussions, including the corollary on consistency and remarks on assumptions (ii)--(v), are given in Supplementary Section~\ref{supp:consistency}. 
Applying Proposition~\ref{prp:bps-weight-concentration} to the Gaussian weight reveals the asymptotic behavior of the posterior of window parameter $\sigma_k$. To cover the boundary, we set $u_k = 1/(1+\sigma_k^2)$, so that the support of $u_k$ is $[0,1]$ and compact, and $\sigma_k^2 = \infty$ and $\sigma_k^2 = 0$ correspond to $u_k=0$ and $u_k=1$, respectively. 
In this parametrization, the limits in Proposition~\ref{prp:bps-weight-concentration} for pseudo truth $u_k^{\star}\in \{ 0, 1\}$ imply the posterior concentration at boundaries. 
\begin{itemize}
    \item {\it Constant weights.} If $u_k^\star=0$ for all $k$, then posterior mass concentrates in a neighborhood of $\sigma_k^2=\infty$. The weight $w_k(x_k;\theta)$ approaches constant $q_k$ at every $x_k$, making the density combination less outcome-dependent. Note also that the baseline receives zero weight in the limit. In total, the limiting predictive density is nothing but the constant-weight linear pool. 
    \item {\it Dominant baseline.} If $u_k^\star=1$ for some $k$, then posterior mass concentrates in a neighborhood of $\sigma_k^2=0$. In this limit, $w_k(x_k;\theta)$ converges to zero for every $x_k\ne\mu_k$. 
    The exception is $x_k=\mu_k$, but this value of $x_k$ has probability zero under a continuous agent density. 
    Hence the component of agent $k$ vanishes in the mixture form in (\ref{eq:diracpred}). If this occurs for all agents, the predictive density reduces to the baseline $h_0$.
\end{itemize}

\subsection{Illustration: a two-part normal example}
\label{sec:gaussian-illustration}

We illustrate the aforementioned behavior of the posterior distribution under the Gaussian weights using a two-part normal example.
Suppose that the data generating process has density
\begin{equation*}
    p^\ast(y)
    =
    \frac{1}{3}\,
    \mathcal{TN}_{(-\infty,0]}(y\mid0,1^2)
    +
    \frac{2}{3}\,
    \mathcal{TN}_{[0,\infty)}(y\mid0,2^2),
\end{equation*}
where $\mathcal{TN}_A$ denotes the truncated normal density whose support is
$A\subset\mathbb{R}$.
The two agent densities are
\[
    h_1(x)=\mathcal{N}(x\mid0,1^2)
    \qquad {\rm and} \qquad 
    h_2(x)=\mathcal{N}(x\mid0,2^2).
\]
It would be ideal if the first agent were used on the left half line and the second agent on the right or, equivalently, if we could set the weight functions to the switching weights considered in Section~\ref{sec:gaussian-weight}:
\[
    w_1^{\mathrm{sw}}(y)=\mathbb{I}(y<0)
    \qquad {\rm and} \qquad 
    w_2^{\mathrm{sw}}(y)=\mathbb{I}(y\ge0).
\]

For the baseline, we consider two specifications: a diffuse baseline,
$h_0(y)=\mathcal{N}(y\mid0,10^2)$,
and a moment-matched baseline,
$h_0(y)=\mathcal{N}(y\mid m_B,s_B^2)$,
whose mean and variance match those of the data generating process $p^\ast$: $m_B=\sqrt{2/\pi}$ and $s_B^2=3-2/\pi$.
In this example, we use the Dirac synthesis to ensure the identification of the Gaussian weight parameters, including $\sigma_k$.

For each baseline
$b\in\{\mathrm{diff},\mathrm{match}\}$,
let $\theta_b^\star$ denote the corresponding pseudo-true value of parameter set
$\theta=(q_{1:2},\mu_{1:2},\sigma_{1:2}^2)$. 
We numerically compute $\theta_b^\star$ by maximizing the population expected log score defined in Proposition~\ref{prp:bps-weight-concentration}. 
Table~\ref{tab:klproj-main} summarizes the computed width parameters $\sigma_k^{2\star}$, the baseline share, and the expected log scores at the computed pseudo-truth $(\sigma_1^{2\star},\sigma_2^{2\star})$ and at the boundary
$(\sigma_1^2,\sigma_2^2)=(\infty,\infty)$.
For the diffuse baseline, we obtain $\sigma_1^{2\star}=\infty$ and $\sigma_2^{2\star}\simeq119$; the KL-projected values are at the boundary. 
The baseline share is about $0.021$, making the baseline negligible. The improvement of the objective function over the case of $(\sigma_1^2,\sigma_2^2)=(\infty,\infty)$ is small; the value of $\sigma_2^{2\star}\simeq119$ should be considered sufficiently large. 
In contrast, for the moment-matched baseline, the values of the pseudo truth are both finite: 
$(\sigma_1^{2\star},\sigma_2^{2\star})\simeq(1.74,6.39)$. 
The baseline is dominant in the combined density with weight $0.576$, and the difference from the case of $(\sigma_1^2,\sigma_2^2)=(\infty,\infty)$ becomes significant. 
For the pseudo-true quantities other than $\sigma_k^2$ , see Table~\ref{tab:supp-klproj-finite} in Section~\ref{supp:klproj}.

\begin{table}[htbp]
\centering
\caption{
Pseudo true values of $\sigma_k^2$, weight of baseline $h_0$, functional values of the objective function $M(\theta)$ at $(\sigma _1^2,\sigma_2^2) = (\sigma_1^{2\star},\sigma_2^{2\star})$ and $(\infty,\infty)$, and their difference. 
}
\label{tab:klproj-main}
\vspace{3pt}
\begin{tabular}{cccccc}
\hline
\multirow{2}{*}{Baseline}
& \multirow{2}{*}{$(\sigma_1^{2\star},\sigma_2^{2\star})$}
& Baseline share
& $\sup M(\theta)$
& $\sup M(\theta)$
& \multirow{2}{*}{Difference} \\
{}
& {}
& ($h_0$'s weight)
& at $(\sigma_1^{2\star},\sigma_2^{2\star})$
& at $\sigma_1^{2}=\sigma_2^{2}=\infty$
& \\
\hline
diffuse
& $(\infty,\ 119.23)$
& $0.0208$
& $-1.957243$
& $-1.958195$
& $9.53\times10^{-4}$ \\
matched
& $(1.736,\ 6.391)$
& $0.5765$
& $-1.828608$
& $-1.958195$
& $0.12959$ \\
\hline
\end{tabular}
\end{table}

The weight functions and predictive distributions at $\theta^{\star}_b$ are shown in Figure~\ref{fig:kl_illustration}.
Under the diffuse baseline, the first agent receives a constant weight and the second weight varies only mildly, so the synthesis remains close to a constant-weight pool.
Under the moment-matched baseline, both weights vary over the outcome space, and the baseline retains a substantial share.
Nevertheless, both baseline specifications fail to recover the desired left--right regional allocation of weights, motivating the development of a new class of weight functions.

\begin{figure}[htbp]
\centering
\includegraphics[width=\linewidth]{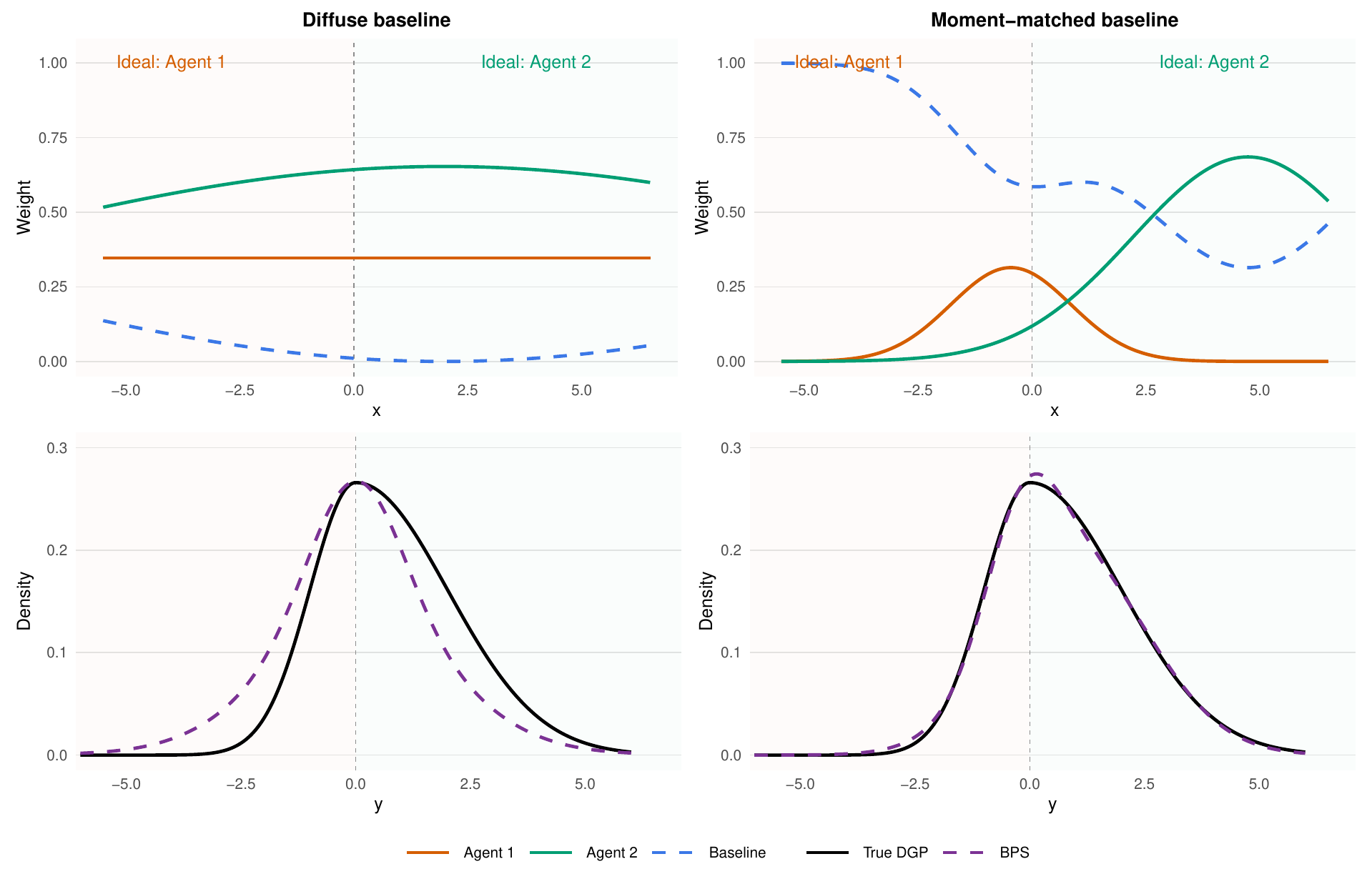}
\caption{
Dirac Gaussian-weights at the computed pseudo true value $\theta_b^{\star}$ (top) and the implied BPS density with the true density (bottom) in the two-part normal example. 
}
\label{fig:kl_illustration}
\end{figure}

\section{Softmax weight models}
\label{sec:softmax}

This section develops a new class of the weight functions. Based on our observation on the Gaussian weights in the previous section, we choose the softmax transformation of the agent scores, so that the weight of agent $k$ is dependent on $x$, or the relative performance of $k$ to the other agents.

\subsection{Softmax weight functions}
\label{sec:softmax-spec}

For each $k=1,\dots,K$, consider agent $k$'s score $\eta_k(x_k;\gamma_k) \in \mathbb{R}$, a function of latent state $x_k$ that is specified later. We denote its parameter by $\gamma_k$ rather than $\theta_k$. For convenience, we define the score of the baseline as $\eta_0\equiv0$. Then, we specify the weight function as the softmax transformation of the agent's scores: 
\begin{equation}
    w_k(x;\gamma)
    = \frac{\exp\{\eta_k(x_k;\gamma_k)\}}
      {1+\sum_{\ell=1}^K \exp\{\eta_\ell(x_\ell;\gamma_\ell)\}},
    \quad k=1,\ldots,K,
    \label{eq:softmax-weight}
\end{equation}
where $\gamma=(\gamma_1,\ldots,\gamma_K)$. We also set $w_0(x;\gamma)=1-\sum_{k=1}^K w_k(x;\gamma)$. The synthesis function under the bias-adjusted synthesis is of the form in (\ref{eq:biasad}), with $\theta$ being replaced with $\gamma$. 
By construction, the weight functions, $w_k(\cdot ; \gamma)$ for $k=0,1,\dots ,K$, are non-negative and sum to unity for any $x$ and $\gamma$. 

Unlike the Gaussian weights, the softmax weights depend on all components of $x$ through their relative scores and do not impose fixed agent-specific caps $q_k$.
A decline in one agent's score reallocates its weight across all competing components based on their relative scores rather than solely to the baseline.
Hence the softmax formulation is expected to be suited to examples such as one in Section~\ref{sec:gaussian-illustration}, where outcome-dependent way of combination of agent forecasts is truly necessary.

\subsection{Models of score functions}
\label{sec:softmax-scores}

The model for the score function, $\eta_k(x_k;\gamma_k)$, is arbitrary. In theory, one can even model it in a nonparametric way, using the basis function expansion in terms of $x_k$. In practice, the flexibility of the score function comes at the cost of inferential efficiency due to the increased number of parameters. Here, we focus on the first and second order polynomial scores: 
\begin{equation*}
    \eta_k(x_k;\gamma_k)
    =
    \gamma_{k0}+\gamma_{k1}x_k+\gamma_{k2}x_k^2 ,
\end{equation*}
where the first-order polynomial score is obtained with $\gamma_{k2}=0$. 
The first-order polynomial model is simply increasing or decreasing in $x_k$, differentiating the importance of agent $k$ in the left and right. The second-order polynomial allows each agent to favor the interior region or tails, realizing center-versus-tail competition of agents in density combination. 

An example of the functional forms of the implied weights is shown in Figure~\ref{fig:weight_illustration} for illustration. The difference between the Gaussian and softmax weights is clearly seen in this figure; the softmax weight is neither bounded, nor limited to a log-concave function. The first-order polynomial model is expected to accurately approximate the data generating function in the example of Section~\ref{sec:gaussian-illustration}.

\begin{figure}[htbp!]
\centering
\includegraphics[width=\linewidth]{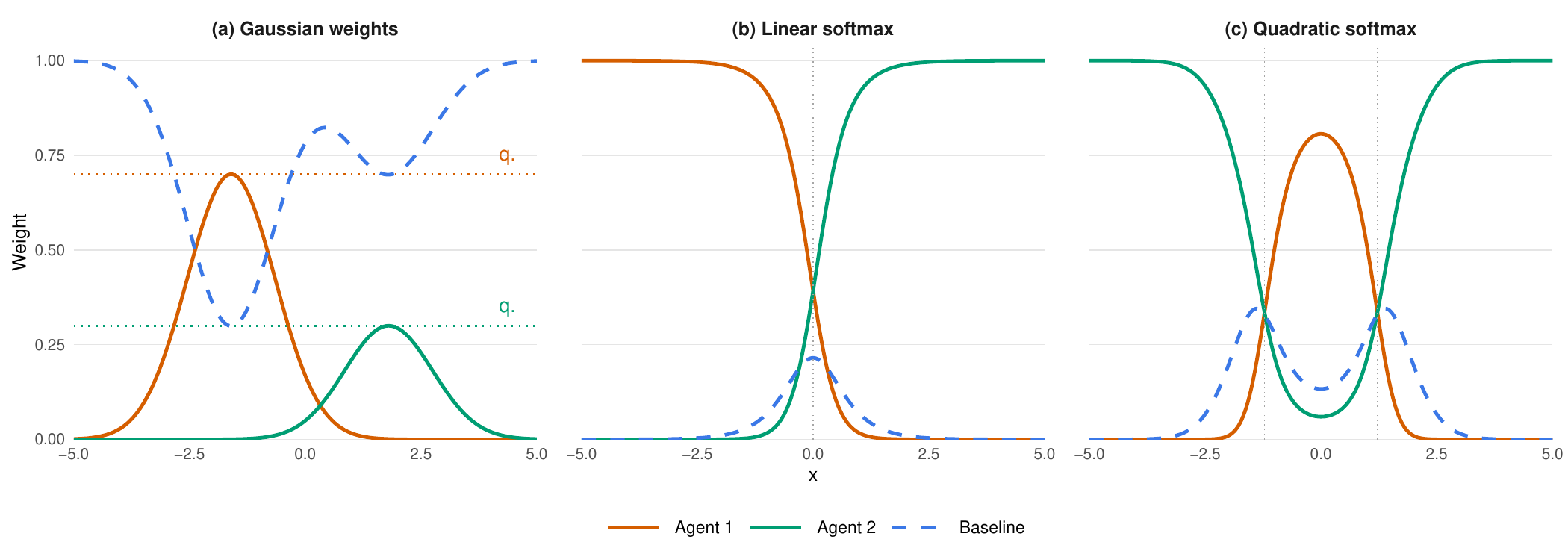} 
\caption{Examples of the weight functions $w_k$ of the two agents (solid) and baseline (dashed) for (a) the Gaussian weights, (b) the Softmax weights with the first-order polynomial score, and (c) that of the second-order polynomial score. In (a), we set $(q_1,q_2)=(0.7,0.3)$, $(\mu_1,\mu_2)=(-1.6,1.8)$, and $\sigma_1^2=\sigma_2^2=0.9$. In (b), the scores are $\eta_1(x)=0.6-2x$ and $\eta_2(x)=0.6+2x$. In (c), we use $\eta_1(x)=1.8-1.2x^2$ and $\eta_2(x)=-0.8+0.55x^2$.
}
\label{fig:weight_illustration}
\end{figure}

\section{Derivation of the MCMC algorithm}
\label{sec:mcmc}

\subsection{MCMC for the Gaussian weights}

An MCMC algorithm for posterior computation under a certain class of mixture BPS models was provided by \cite{johnsonwest2025outcome}. In their model, the synthesis function is of the Dirac form with the Gaussian weights. With an additional latent variable $z \in \{ 0,1,\dots ,K\}$, the synthesized predictive distribution in (\ref{eq:diracpred}) has the compositional form as 
\begin{equation*}
    \begin{split}
        &y|x,\mathcal{H},z{=}k \sim \delta _{x_k}(y) \qquad (k\not=0), \\
        &y|x,\mathcal{H},z{=}0 \sim h_0(y), \\
        &z|x,\theta \sim \mathrm{Categorical}( w_0(x;\theta),w_1(x_1;\theta),\dots , w_K(x_K;\theta) ), \\
        &x \sim h(x).
    \end{split}
\end{equation*}
The deterministic relationship between $y$ and $x$ in the degenerate Dirac distributions complicates the construction of valid Markov chains. In particular, conditional that $z=k$, the Dirac component imposes the deterministic constraint $y=x_k$. Conversely, given $x=(x_1,\dots,x_K)$ such that $y=x_k$, the conditional distribution of $z$ is degenerated at $z=k$. As a result, the iterative sampling from the full conditionals does not constitute a Markov chain converging to the target posterior distribution. 
To circumvent this difficulty, \cite{johnsonwest2025outcome} marginalized $x$ out when sampling from the full conditional of $z$. While this approach is expected to improve the efficiency of the MCMC method, the marginalization of $x$ requires the use of Monte Carlo integration {\it within} the iterations of MCMC, which is often computationally costly.

To avoid the need for nested Monte Carlo integration caused by the deterministic constraint of the Dirac synthesis, we instead work with the bias-adjusted synthesis in (\ref{eq:biasad}), using latent allocation variable $z$. 
Conditional on $z=k$ ($k=1,\ldots,K$), the observational equation becomes
\begin{equation*}
    y\mid x,\mathcal{H},z{=}k
    \sim
    \mathcal{N}(y\mid x_k-\beta_k,\tau_k^2).
\end{equation*}
With a prior for $(\beta_k,\tau_k^2)$, the full conditional distribution of $x$ is no longer degenerate, and that of $z$ is supported on $\{0,1,\ldots,K\}$. The MCMC algorithm for the Gaussian weights can be constructed easily by combining Gibbs and Metropolis--Hastings updates. For the full details of the algorithm, see Section~\ref{sec:supp-mcmc-gaussian}.

\subsection{MCMC for the softmax weights}

Unlike the Gaussian weight, the softmax weight is complex as a function of $x$ and $\gamma$, and no prior is known to be conditionally conjugate. This problem is also seen in any model with the logistic-normal likelihood, such as the logistic and multinomial regressions. For this class of models, \cite{polson2013bayesian} proposed the parameter augmentation using the P\'olya-gamma distribution to allow the conditionally conjugate analysis and Gibbs sampler, which is directly applicable to the mixture BPS with the softmax synthesis. 

Consider the full conditional posterior distribution of $\gamma_k$ for some $k$, conditioning all the other variables. 
For simplicity, we first consider a single observation $y$ ($n=1$).
To read off the logistic-normal likelihood, we set 
\[
    \psi_{k}
    =
    \eta_{k}(x_k;\gamma_k)
    - \log\left\{
      \sum_{\ell\ne k}
      \exp(\eta_{\ell}(x_{\ell};\gamma_{\ell}))
    \right\},
\]
where the sum in the right includes $\ell = 0$ and assumes $\eta_0=0$. Writing $d_{k}=\mathbb I(z=k)$, the density of $p(z|x,\gamma)$ can be expressed as the binary-logit form up to proportionality,
\[
    p(z|x,\gamma) \propto
    \frac{\{\exp(\psi_k)\}^{d_k}}
         {1+\exp(\psi_k)}.
\]
The P\'olya-gamma augmentation implies that, with an additional latent variable $\omega _k>0$, the full conditional of $\gamma_k$ is proportional to 
\begin{equation*} 
   \Pi (\gamma_k) \exp\left\{
      \left( d_{k}-\frac12 \right) \psi_{k}
      -\frac{\omega_{k}}{2}\psi_{k}^2
    \right\},
\end{equation*}
The exponent is quadratic in $\psi _k$ (or $\eta_k$), then the normal prior becomes conjugate. If the score is the first-order polynomial, $\psi_k$ is linear in $\gamma_k$ and $x_k$, hence the normal priors for $\gamma _k$ and $x_k$ are conjugate. For the second-order polynomial score, where $\psi_k$ is quadratic in $x_k$, we update $x_k$ by a Metropolis--Hastings method with the P\'olya--Gamma variables being marginalized out, and subsequently resample the P\'olya--Gamma variables before updating $\gamma_k$. 
The full conditional of $\omega _k$ is a P\'olya-gamma distribution, the simulation from which is known to be feasible \citep{polson2013bayesian}. For general $n$, we have $n$ latent variables and exponential factors in the expression above.

The other parameters, $\beta$ and $\tau^2$, have conjugate priors and can be sampled directly. The algorithms of the proposed MCMC method for the linear and quadratic softmax specifications are given in Section~\ref{sec:supp-mcmc-softmax}.

\section{Numerical illustrations}
\label{sec:numerical}

Using simulated datasets, we evaluate the out-of-sample performance in density estimation of the mixture BPS with the Gaussian and softmax weights, and compare them with those of the existing methods.

\paragraph{Experimental design} 
The experiment is a two-step program. First, we need to estimate the parameters in each method (e.g., $\theta$ and $(\beta , \tau^2)$ in the mixture BPS), using the test data of size $n_{\rm test}$. Note that this is the ``test'' data for the agent forecast densities, but this can also be viewed as the ``training'' data for the synthesis function, or the method for density combination. Second, we simulate additional data for evaluation of size $n_{\rm eval}$ to compute measures of predictive performance. In each experiment, we set $n_{\rm test} \in\{10,100,1000\}$ and $n_{\rm eval}=10000$.

\paragraph{Data generating processes and agent densities}

We consider two scenarios for data generating processes and agent densities. In both cases, the baseline is set to be diffuse: $h_0(y)=\mathcal{N}(y\mid 0,10^2).$

\begin{itemize}
    \item DGP-A: Center-versus-tail switching.
    
The first experiment focuses on the center and tails of the agent densities. The data generating process is 
\begin{align*}
  p_A^\ast(y)
  =
  \frac{1}{2}\mathcal{TN}_{[-2,2]}(y\mid 0,1^2)
  +
  \frac{1}{2}\mathcal{TN}_{|y|>2}(y\mid 0,3^2),
\end{align*}
while agent densities are
$h_1(x)=\mathcal{N}(x\mid0,1^2)$ and
$h_2(x)=\mathcal{N}(x\mid0,3^2)$. Except for the truncation of the support, the agent densities correctly specify the two components of the data generating process. Thus, the weight of $h_1$ in the mixture BPS should be one for $|y|\le2$, and zero otherwise. Similarly, the weight of $h_2$ should be one for $|y|>2$. Finally, note that $p_A^{\ast}$ is discontinuous at $y=\pm 2$, as visually confirmed later in Figure~\ref{fig:dgp-a}.

\item DGP-B: one-threshold switching with misspecified agents.

The second experiment is similar to the example in Section~\ref{sec:gaussian-illustration}. The data generating process is the two-part normal density defined by
\begin{align*}
  p_B^\ast(y)
  =
  \frac{1}{3}\mathcal{TN}_{(-\infty,0]}(y\mid0,1^2)
  +
  \frac{2}{3}\mathcal{TN}_{[0,\infty)}(y\mid0,2^2).
\end{align*}
Unlike the previous experiment, the agent densities are assumed to be misspecified. Specifically, we use
\[
  h_1(x)=\mathcal{N}(x\mid0.50,0.80^2),
  \quad {\rm and} \quad
  h_2(x)=\mathcal{N}(x\mid-0.50,1.60^2).
\]
Ideally, we should allocate the whole weight to agent~1 on the negative half of the real line, and to agent~2 on the positive half, while correcting the misspecification.

\end{itemize}

\paragraph{Competing methods and evaluation}

In addition to the mixture BPS, we consider other methods for density combination listed below. 
\begin{itemize}
  \item \textbf{Equal-weight pool (EW)}: an equal-weight linear pool of the baseline, $h_1$ and $h_2$. 
  \item \textbf{Linear pool (LP)}: a constant-weight linear pool whose weights are optimized using the training-sample log score \citep{hallmitchell2007combining}. 
  \item \textbf{Generalized density forecast combination (GDFC)}: the generalized density forecast combination of \citet{kapetanios2015generalised}, using piecewise-constant outcome-dependent coefficients. This approach can be viewed as a mixture BPS, where the weights are modeled nonparametrically by the histogram based on the pointwise-normalized region coefficients. 
  \item \textbf{Gaussian weight with fixed window (BPS-G0)}: The mixture BPS with Gaussian-weights, whose window parameters are fixed at $\sigma_1^2=\sigma_2^2=1$;
  \item \textbf{Gaussian weight  (BPS-G)}: The mixture BPS with Gaussian-weights, whose parameters are estimated;
  \item \textbf{Linear softmax weight (BPS-S1)}: The mixture BPS with softmax weights based on the first-order polynomial score ($\gamma_{k2}=0$ for all $k$);
  \item \textbf{Quadratic softmax weight (BPS-S2)}: The score is the second-order polynomial.
\end{itemize}

For the mixture BPS, each Markov chain is run for 100,000 iterations, with the first 25,000 discarded. The prior for $(\beta ,\tau^2)$ is informative for identification as discussed in Section~\ref{sec:syngen}, while those for $\theta$ (or $\gamma$) are sufficiently diffuse. 
For the details of the listed method, including the model specification, prior settings and estimation algorithms, see Section~\ref{supp:num}.

\paragraph{Evaluation measures}

The posterior of the outcome-dependent weight functions ($w_k(\cdot ;\theta)$ in the mixture BPS) is of our primary interest. By drawing their functional forms, we will graphically confirm whether the ideal weight allocation discussed before is achieved or not.

In terms of the out-of-sample fit, predictive performance is primarily
evaluated using the average log score and the continuous ranked
probability score (CRPS).
For an independent test sample,
\[
  \mathrm{LogS}
  =
  \frac{1}{n_{\mathrm{eval}}}
  \sum_{i=1}^{n_{\mathrm{eval}}}
  \log \widehat p(y_i^{\mathrm{eval}}),
\]
where $\widehat p$ denotes the fitted predictive density. 
The average CRPS is
\[
  \mathrm{CRPS}
  =
  \frac{1}{n_{\mathrm{eval}}}
  \sum_{i=1}^{n_{\mathrm{eval}}}
  \int_{-\infty}^{\infty}
  \left\{
    \widehat F(z)
    -
    \mathbb{I}(y_i^{\mathrm{eval}}\le z)
  \right\}^{2} dz ,
\]
where $\widehat F$ is the predictive distribution function corresponding to $\hat{p}$. 
The larger log scores and smaller CRPS values indicate the better predictive performance.

\paragraph{Results}

Figures~\ref{fig:dgp-a} and \ref{fig:dgp-b} show the point estimates/posterior distributions of the weight functions and the true/estimated predictive densities for DGP-A and DGP-B, respectively, when $n_{\rm test}=1000$. In the center-versus-tails scenario of DGP-A, only GDFC and BPS-S2 show the estimates of the weight functions as desired, using agent~1 in the center and agent~2 in the tails. For all the methods, it is difficult to accurately estimate the density kinks of the data generating process at $y=\pm 2$ even though we set $n_{\rm test}$ to be moderately large. The nonparametric approach of GDFC estimates the right kink correctly, but fails to do so in the left. The BPS-S2, in contrast, provides a wider credible intervals, showing its honest quantification of uncertainty. This posterior uncertainty leads to the smoothed density estimates around $y=\pm 2$. 

In the result of DGP-B, BPS-S1 shows the ideal switching of weights from agent~1 on the left to agent~2 on the right in its posterior plot, which is consistent with the asymptotic result illustrated in Section~\ref{sec:gaussian-illustration}, despite the misspecification of the agent densities. The change point is slightly shifted from the origin to the left, but this can be explained by the bias adjustment in the agents' locations; the posterior means of $(\beta_1,\beta_2)$ are $(-1.06,0.00)$. 
Similarly, we can confirm the switch from agent 1 to agent 2 under BPS-S2, except in the right tail that does not contribute to the overall density estimate (the estimates of the location biases under BPS-S2 are $(\beta_1,\beta_2)=(-0.92,-0.99)$). 
Under GDFC, the point estimates of the weights are roughly the same as those of BPS-S1, but the resulting density estimate is discontinuous at several points due to the nature of the nonparametric approach.

Under the mixture BPS with Gaussian weights, the posterior weights are close to constant functions in both experiments when estimating window parameter $\sigma_k$ (BPS-G). This result can be predicted from the asymptotic result and discussions given in Section~\ref{sec:gaussian-asymptotics}. With the window parameters fixed to $\sigma_k=1$ (BPS-G0), the weights become more outcome-dependent, but their functional form is constrained to be quadratic and, as a result, the density tails are explained by the baseline, not agent 2.

Table~\ref{tab:main-performance} reports the log scores and CRPSs for each $n_{\rm test}$. In most experiments, the mixture BPS with the softmax weights (BPS-S1 and BPS-S2) is the best method in terms of the out-of-sample fit. In DGP-A with small samples, LP can be the best method due to its simple structure. However, LP's simplicity, in turn, limits its flexibility in combining densities with a moderate/sufficient amount of test data, as seen in its estimates of predictive densities in Figures~\ref{fig:dgp-a} and \ref{fig:dgp-b}. 

Taken together, the experiments illustrate the importance of flexibility in the weight function: the Gaussian weights tend toward nearly constant weights when their window parameters are estimated, whereas the softmax weights can represent the regional switching.

\begin{figure}[htbp]
\centering
\includegraphics[width=0.75\textwidth]
{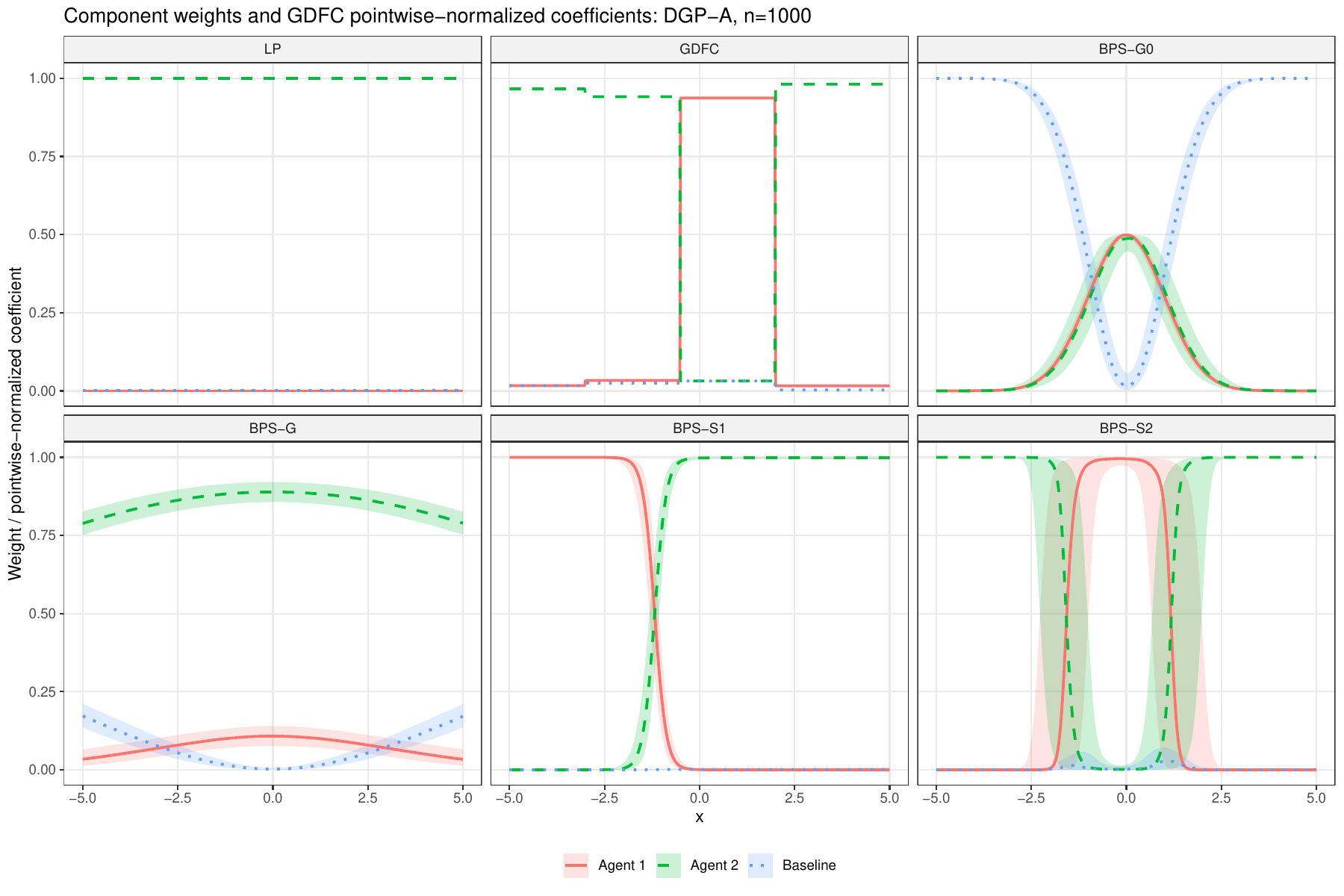}\\
\includegraphics[width=0.75\textwidth]
{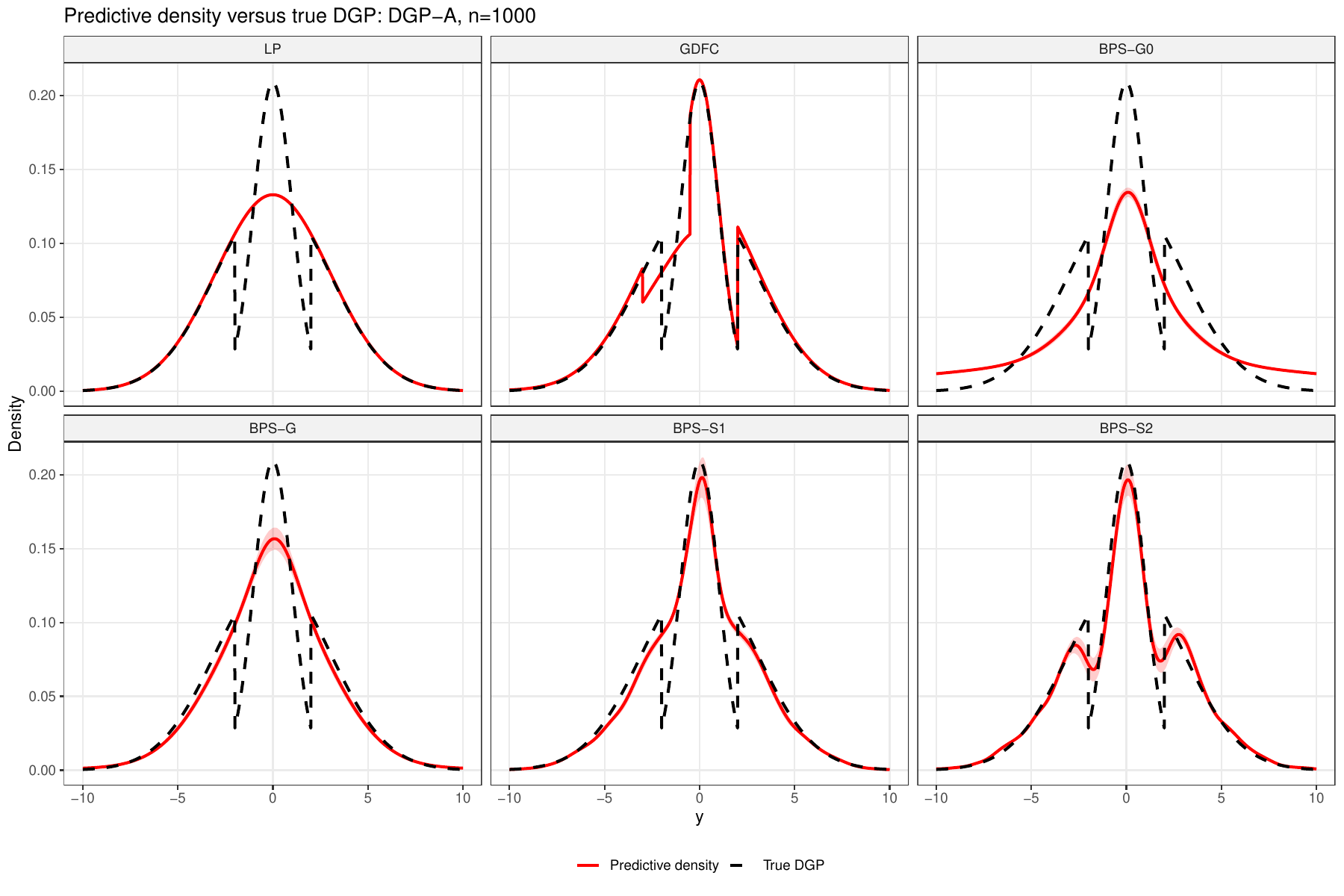}
\caption{
The estimated weights (top) and densities (bottom) for DGP-A with $n=1000$. Under the mixture BPS, the weight curves are posterior medians and the shaded regions are the 15th--85th posterior quantile bands.
For the predictive densities, the red solid curves are posterior means of the conditional predictive densities and the red shaded regions are the corresponding 15th--85th pointwise posterior credible bands.
}
\label{fig:dgp-a}
\end{figure}

\begin{figure}[htbp]
\centering
\includegraphics[width=0.75\textwidth]
{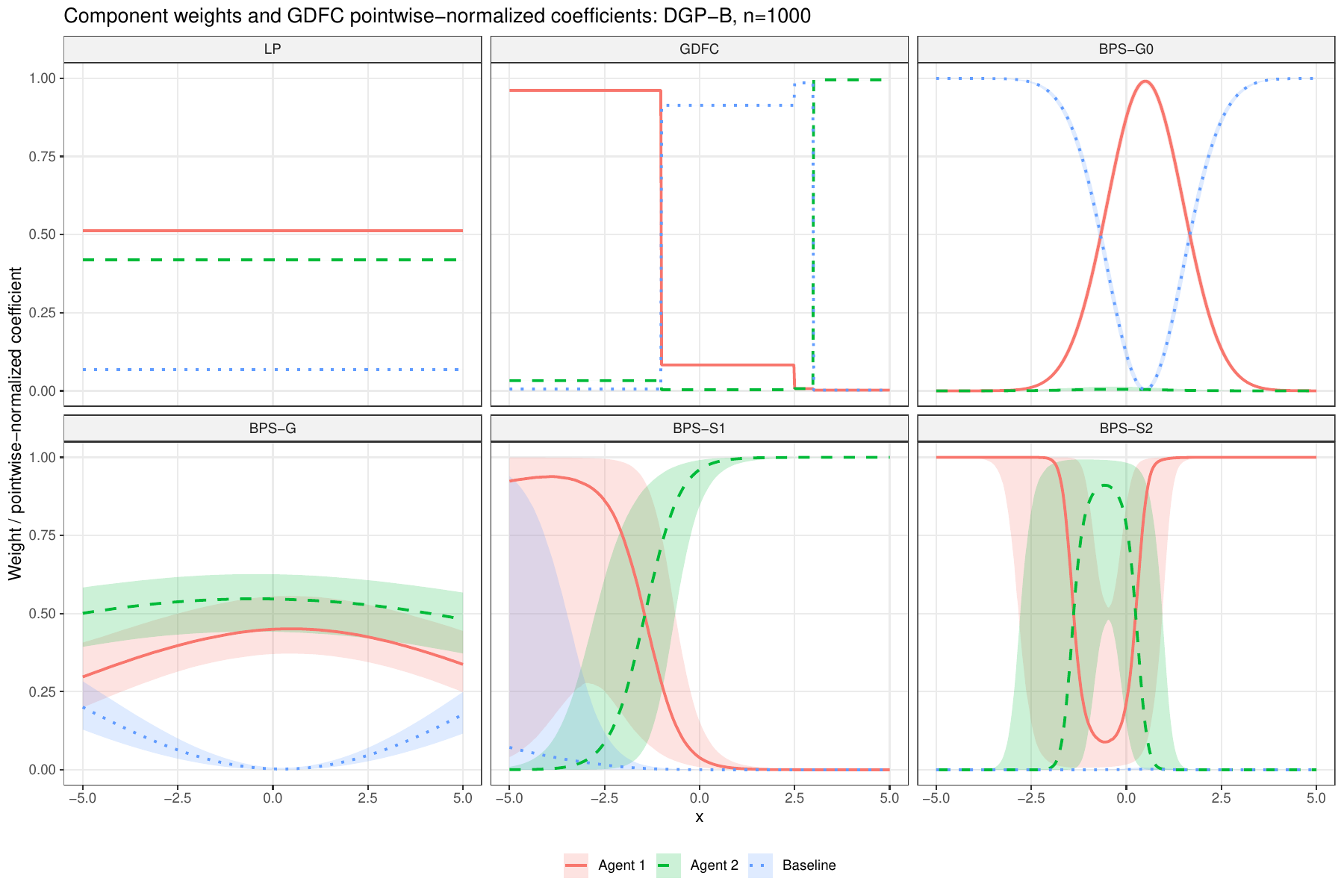} \\
\includegraphics[width=0.75\textwidth]
{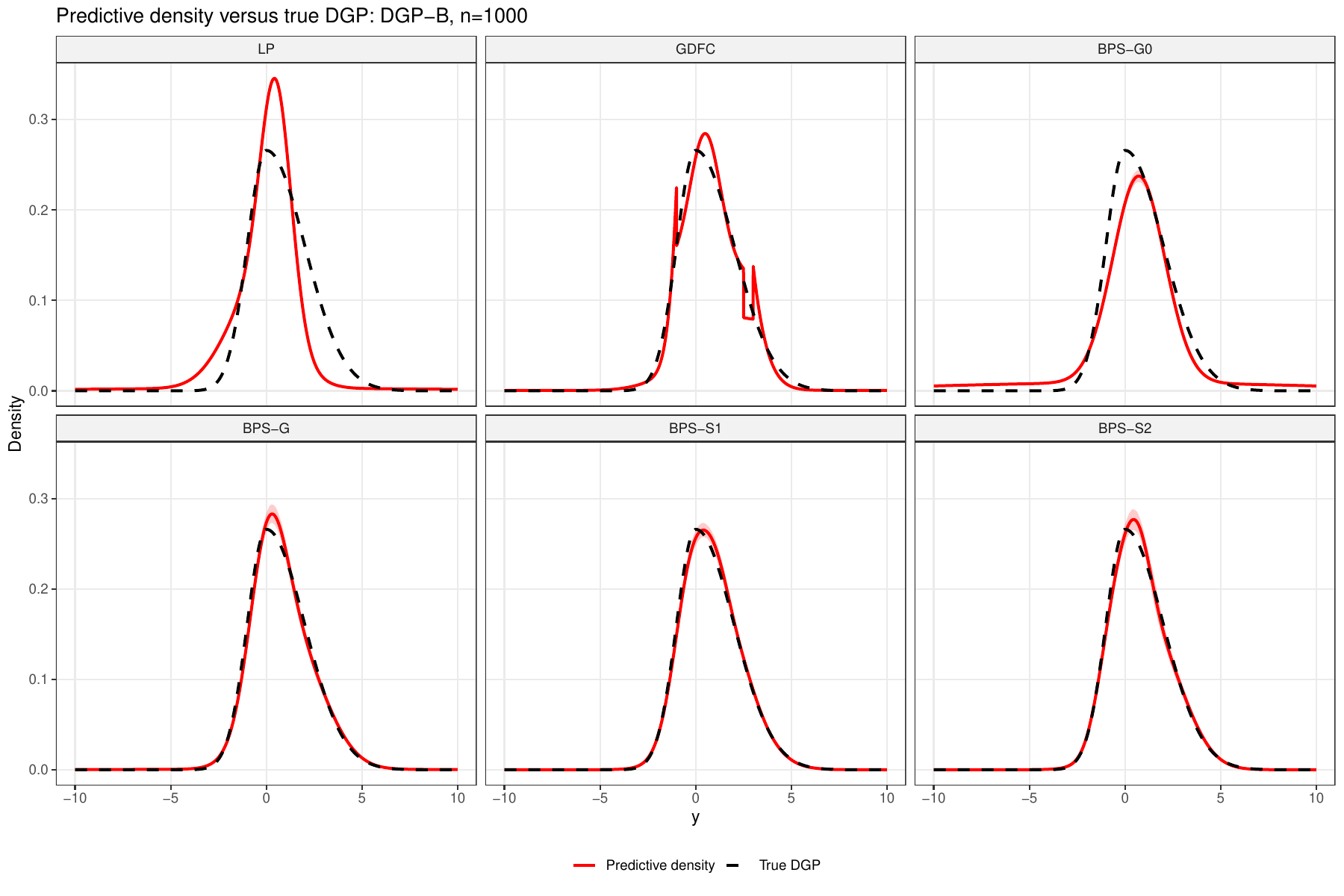}
\caption{
Estimated component weights (top) and predictive densities (bottom) for DGP-B with $n=1000$.
For the Bayesian BPS methods, the weight curves are posterior medians and the shaded regions are the 15th--85th posterior quantile bands.
For the predictive densities, the red solid curves are posterior means of the conditional predictive densities and the red shaded regions are the corresponding 15th--85th pointwise posterior credible bands.
}
\label{fig:dgp-b}
\end{figure}

\begin{table}[htbp]
\centering
\small
\setlength{\tabcolsep}{4.5pt}
\caption{
The average log scores and continuous ranked probability scores of: equal-weight pool (EW), linear pool (LP), generalized density forecast combination (GDFC), mixture BPS with Gaussian weights with the window parameters being fixed (BPS-G0) and estimated (BPS-G) and softmax weights with the first-order (BPS-S1) and second-order (BPS-S2) polynomial scores. 
The sample size in computing those measures is $n_{\mathrm{eval}}=10{,}000$ in all cases.}
\label{tab:main-performance}

\vspace{0.5em}

\begin{tabular}{lrrrrrr}
\toprule
& \multicolumn{2}{c}{$n_{\rm test}=10$}
& \multicolumn{2}{c}{$n_{\rm test}=100$}
& \multicolumn{2}{c}{$n_{\rm test}=1000$} \\
\cmidrule(lr){2-3}
\cmidrule(lr){4-5}
\cmidrule(lr){6-7}
Method
& LogS & CRPS
& LogS & CRPS
& LogS & CRPS \\
\midrule

\multicolumn{7}{l}{\textit{DGP-A}} \\
EW & -2.6509 & 1.7240 & -2.6444 & 1.7172 & -2.6447 & 1.7209 \\
LP & \textbf{-2.5000} & \textbf{1.6548}
   & -2.4859 & \textbf{1.6435}
   & -2.4988 & 1.6486 \\
GDFC & -3.0476 & 1.6651
   & -2.4975 & 1.6501
   & -2.4726 & 1.6458 \\
BPS-G0 & -2.8180 & 1.9440
   & -2.6605 & 1.7254
   & -2.7121 & 1.8507 \\
BPS-G & -2.6264 & 1.6930
   & -2.5293 & 1.6508
   & -2.5010 & 1.6467 \\
BPS-S1 & -2.5628 & 1.6705
   & -2.4945 & 1.6499
   & -2.4767 & \textbf{1.6455} \\
BPS-S2 & -2.5313 & 1.6667
   & \textbf{-2.4822} & 1.6472
   & \textbf{-2.4687} & 1.6470 \\

\midrule

\multicolumn{7}{l}{\textit{DGP-B}} \\
EW & -2.2236 & 1.1041
   & -2.2064 & 1.0889
   & -2.1912 & 1.0869 \\
LP & -2.2093 & 1.0354
   & -2.0973 & 0.9606
   & -2.0687 & 0.9382 \\
GDFC & -2.0795 & 0.9818
   & -1.9072 & 0.8759
   & -1.8460 & 0.8565 \\
BPS-G0 & -2.2421 & 1.2031
   & -2.0447 & 0.9640
   & -1.9853 & 0.9263 \\
BPS-G & -1.9947 & 0.9406
   & -1.9047 & 0.8812
   & -1.8284 & 0.8561 \\
BPS-S1 & -1.9599 & 0.9123
   & \textbf{-1.8679} & 0.8753
   & \textbf{-1.8211} & 0.8563 \\
BPS-S2 & \textbf{-1.9480} & \textbf{0.9076}
   & -1.8685 & \textbf{0.8735}
   & -1.8231 & \textbf{0.8560} \\

\bottomrule
\end{tabular}

\vspace{0.5em}

\begin{minipage}{0.92\textwidth}
\footnotesize
Note: Larger log scores and smaller CRPS values indicate better predictive
performance. Bold entries indicate the best value within each DGP and
training-sample size. Each design cell is based on one replication.
\end{minipage}
\end{table}

\section{Application: quantile forecast combination}
\label{sec:application}

As an illustration of the utility of the proposed method, we apply the proposed BPS-based method to the combination of subjective quantile forecasts \citep{West1992d}. 
We use the revenue data at two movie theatres (the Brattle Theatre and Harvard Square Theatre) reported by \citet{pratt1985repetitive}. Their dataset contains five point forecasts of revenues provided by a single expert for each screening program. 
Since only the point forecasts are reported without any quantification of their uncertainty, the agent forecast distributions in this example are degenerate. 
As discussed in Section~\ref{sec:syngen}, the proposed method of density combination is applicable to this case of degenerate agent ``densities,'' integrating our own uncertainty on such point forecasts via calibration parameters. 

\subsection{Data and forecasting framework}

For each program, the expert reported their subjective estimates of the 25\%, 50\% and 75\% quartiles of gross revenue, together with two extreme values described as ``astonishingly low'' and ``astonishingly high.'' 
Following the treatment in \cite{pratt1985repetitive}, we view the latter two values as the 1\% and 99\% quantiles, respectively. The five agent forecasts are then indexed by $\mathcal R=\{0.01,0.25,0.50,0.75,0.99\}$. 
The reported 100$r$\% quantile of program $i$ is denoted by $\hat{y}_{ir}$, while the realized revenue is $y_i$; both are in log scale and per day. 

The two theatres are analyzed separately. Following the original study, we use the post-summer observations, beginning with Program~25, for estimation and evaluation. 
In program order, the test (training) and evaluation samples contain 52 and 26 observations for Brattle, and 42 and 21 for Harvard Square, respectively. More details on the dataset are given in Section~\ref{supp:app-data}.

\subsection{Agent forecasts and synthesis model}

In this example, the agent forecast distribution is defined as $h_{ir}(x_{ir})=\delta_{\hat{y}_{ir}}(x_{ir})$. The bias-adjusted synthesis for these degenerate forecasts is given in (\ref{eq:deg_ba}) and is applicable to the current example as-is, by replacing the agent index $k\in \{ 1,\dots, K \}$ with $r \in \mathcal{R}$. Note that, in this example, we always set $x_{it} = \hat{y}_{ir}$, hence the latent variables are not sampled but fixed in the posterior computation. The set of parameters to be sampled is $(\theta,\beta,\tau^2)$. In the softmax specification we replace $\theta$ with $\gamma$.

We compare Gaussian weights (BPS-G), linear softmax weights (BPS-S1), and quadratic softmax weights (BPS-S2). 
The baseline $h_0$ is a normal density fitted to each theatre's pre-summer outcomes, Programs~1--24, that are not used in the test and evaluation samples. 
The prior for the weight parameters is constructed carefully, reflecting our prior belief that the $r$-th agent forecast $\hat{y}_{ir}$ should be used more in the neighborhood of the corresponding 100$r$\% quantile of $y_i$. 
For more details, including the hyperparameters of the prior distributions, see Section~\ref{supp:app-priors}.

As a benchmark, we consider the prediction by the least-squares quantile matching (QM), where we construct a normal distribution as the predictive distribution for each program by matching the five quantiles of the normal distribution to those of the expert
\citep{adrian2019vulnerable}. 
Details of the implementation of the QM method are found in Section~\ref{supp:app-evaluation}. 
The predictive performances of the QM and BPS are evaluated by the LogS and CRPS.

\subsection{Estimated weights and predictive densities}
\label{subsec:application-results}

Figure~\ref{fig:app-weights} displays the estimated weight as a function of the outcome. The estimates of the weight functions obtained under the BPS-S2 model is interpretable; for each $r$, the weight for $\hat{y}_{ir}$ is high around the 100$r$\% quantile. The BPS-G provides similar estimates of the weight functions, but most of the total weight is dominated by the baseline density, due to the limitation of the upper-bounded weight functions, as discussed before. In addition, the weights of the three quantiles are all similar, especially in the application to Brattle Theatre. The BPS-S1 provides constant weight estimates and is less outcome-dependent due to the lack of flexibility of the functional form of the score.

\begin{figure}[htbp]
\centering
\includegraphics[width=0.96\linewidth]
{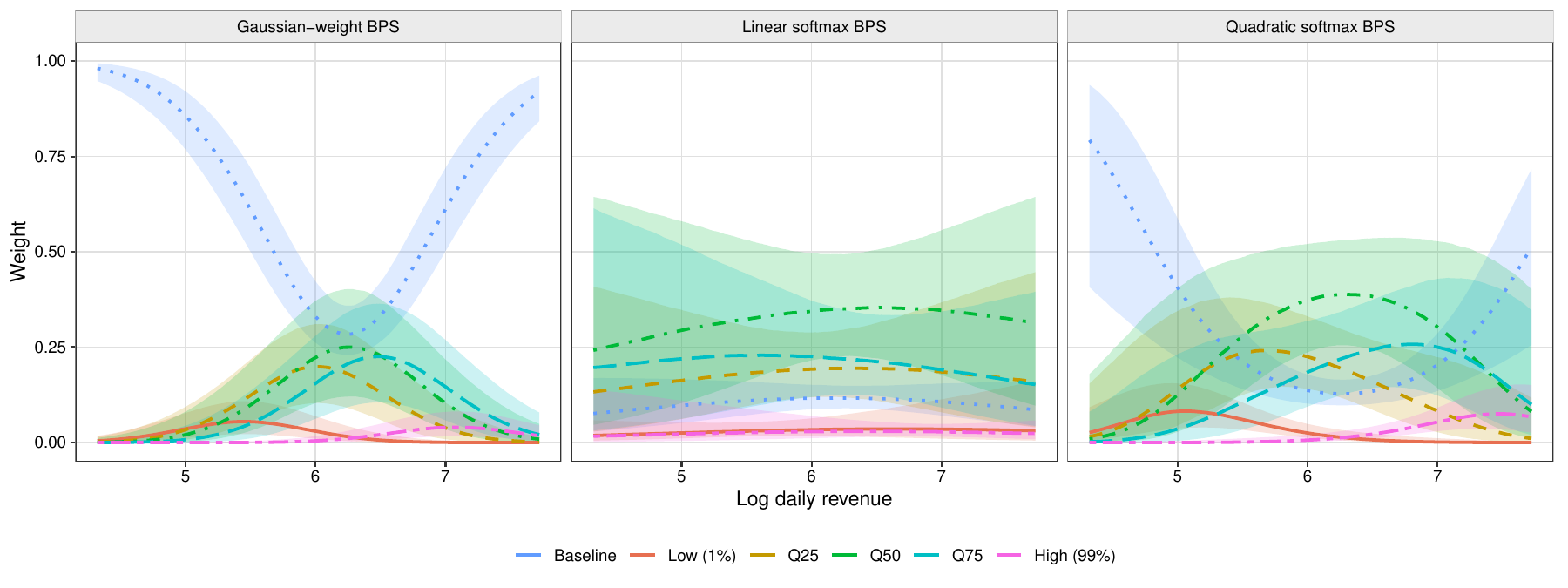}\\[0.6em]
\includegraphics[width=0.96\linewidth]
{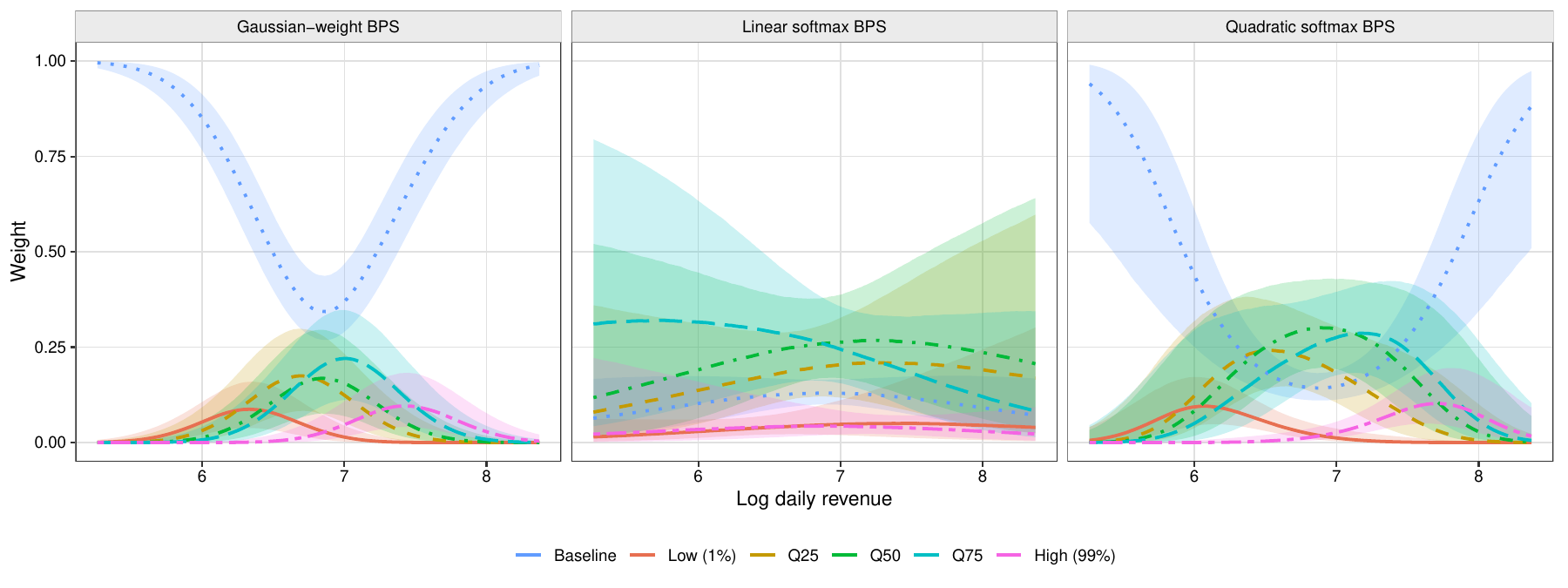}
\caption{Estimated weight functions for Brattle Theatre (top) and
Harvard Square Theatre (bottom), evaluated along
$x_{0.01}=x_{0.25}=x_{0.50}=x_{0.75}=x_{0.99}=g$.
The horizontal axis is $g$ on the log daily revenue scale.
Curves are posterior medians and shaded regions are the
15th--85th pointwise posterior quantile bands.}
\label{fig:app-weights}
\end{figure}

Using the estimated weight functions, each method defines a predictive density customized for program $i$. As an example of such predictive densities, we show the predictive densities of the four methods for Brattle Program~77, the first observation in its evaluation sample, in Figure~\ref{fig:app-density-brattle}. 
The densities obtained under the three BPS-based methods are high around the three quantiles, while leaving some probability mass in the tails. The densities are also skewed to the left, shifting their modes toward the actual observation and improving its predictive accuracy. In addition, the uncertainty on the functional values of the predictive densities is quantified by using the posterior samples; the 70\% credible intervals are shown as the shaded areas in the figure. These properties cannot be seen in the result of the QM. 
BPS-S1 gives the highest log density, $0.1355$, while BPS-S2 gives the smallest CRPS, $0.1091$, for this particular program. 
We have confirmed similar results for other programs; for the example of Harvard Square Program~68, see Section~\ref{supp:app-additional}.

\begin{figure}[htbp]
\centering
\includegraphics[width=0.88\linewidth]
{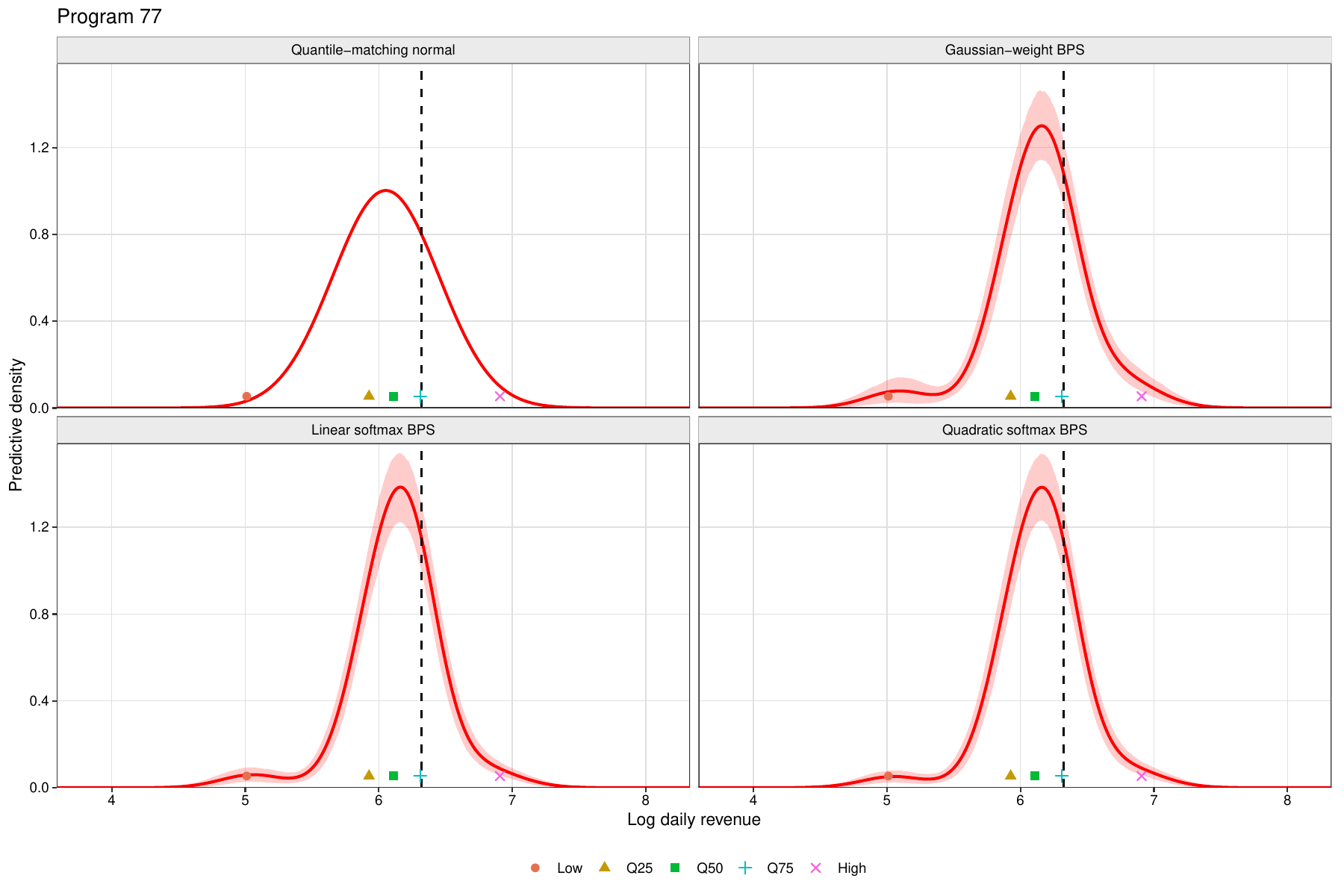}
\caption{Predictive densities for Brattle Theatre, Program~77.
For BPS, red curves are posterior means of the conditional predictive
densities and shaded regions are the 15th--85th posterior
quantile bands. 
The five markers indicate the reported quantile forecasts $\hat{y}_{ir}$, and the vertical dashed
line marks the realized log daily revenue $y_i$.}
\label{fig:app-density-brattle}
\end{figure}

Lastly, we averaged the values of LogS and CRPS across all programs to evaluate the overall predictive performance of the four models, as reported in Table~\ref{tab:app-performance}. 
In both theatres, and in both predictive measures, the best model is the BPS with softmax weights.

\begin{table}[htbp]
\centering
\caption{Average predictive scores for the theatre revenue data.
Larger LogS and smaller CRPS indicate better performance.
Bold entries identify the best value for each theatre and score.}
\label{tab:app-performance}
\begin{tabular}{lrrrr}
\toprule
& \multicolumn{2}{c}{Brattle Theatre}
& \multicolumn{2}{c}{Harvard Square Theatre}\\
\cmidrule(lr){2-3}\cmidrule(lr){4-5}
Method & LogS & CRPS & LogS & CRPS\\
\midrule
QM     & -0.0594 & 0.1410 & -0.3250 & 0.1669\\
BPS-G  & -0.0878 & 0.1403 & -0.2626 & 0.1708\\
BPS-S1 & \textbf{-0.0423} & 0.1378 & -0.2157 & 0.1662\\
BPS-S2 & -0.0467 & \textbf{0.1376}
       & \textbf{-0.2128} & \textbf{0.1654}\\
\bottomrule
\end{tabular}
\end{table}

\section{Concluding remarks and future research}

The softmax weights with the second-order polynomial scores performed well both in the simulation and real data analyses in this study. The success of our analysis critically depends on the choice of priors for $\gamma$, the coefficients of the polynomial. While we subjectively, hence carefully, elicited the hyperparameters based on the context of the data analysis in our application, trying several choices of the hyperparameters and monitoring their performance in LogS and CRPS are also practiced routinely. For higher-order polynomials, we expect that the use of shrinkage prior for $\gamma$ would be useful for an efficient inference. 

Our analysis is limited to the univariate density function with the independently and identically distributed observations. Extension to the multivariate case is a challenging problem in general. The softmax weight can be used even in multivariate settings, but its flexibility should be further investigated. For the proposed model to be applied to a wider range of econometric problems, spatial and/or temporal versions should be developed in the future research.

\section*{Acknowledgment}

The authors thank Katsumi Shimotsu and Mike West for comments and suggestions on the early version of this manuscript. Research of the second author is partially supported by JSPS KAKENHI grant numbers 22K13374. 


\bibliographystyle{chicago}
\bibliography{bps}

\newpage
\setcounter{equation}{0}
\renewcommand{\theequation}{S\arabic{equation}}
\setcounter{section}{0}
\renewcommand{\thesection}{S\arabic{section}}
\setcounter{table}{0}
\renewcommand{\thetable}{S\arabic{table}}
\setcounter{figure}{0}
\renewcommand{\thefigure}{S\arabic{figure}}

\setcounter{page}{1}

\begin{center}
{\LARGE {\bf 
Supplementary Materials of: \\``Outcome-dependent density combination by mixture Bayesian predictive synthesis''
}}
\end{center}

\bigskip
This document collects the proofs, algorithms, details of posterior computation, additional numerical results, and discussions that are omitted in the main text. Contents covered in this documents include: derivation of the marginal predictive density (Section~\ref{supp:marginal}), the proof of Theorem~1 and discussions (Section~\ref{supp:consistency}), details of the numerical computation of the pseudo truth in the two-part normal example (Section~\ref{supp:klproj}), derivation of the MCMC algorithm (Section~\ref{supp:mcmc}), and the full details of the simulation and real data analyses (Sections~\ref{supp:num} and \ref{supp:application}).

\section{Marginal densities under the Bias-adjusted synthesis with Gaussian weights}
\label{supp:marginal}

Below, we derive the marginal density $p(y\mid\mathcal{H},\theta)$ under the Bias-adjusted synthesis with the Gaussian weights defined in Section~2. 

Equation~(\ref{eq:diracpred}) in the main text is obtained by marginalizing the latent states $x$ out. We rewrite this expression as 
\begin{equation}
    p(y\mid\mathcal{H},\theta)
    = c_0(\theta) h_0(y) + \sum_{k=1}^K c_k(\theta) h_k^\ast(y;\theta),
    \label{eq:gauss-marginal-generic}
\end{equation}
where
\begin{equation*}
    c_k(\theta) = \int w_k(u;\theta) h_k(u)\,du,
    \quad
    h_k^\ast(y;\theta) = \frac{w_k(y;\theta) h_k(y)}{c_k(\theta)},
    \quad \mathrm{and} \quad
    c_0(\theta) = 1 - \sum_{k=1}^K c_k(\theta),
\end{equation*}
assuming that $c_k(\theta )>0$ for almost all $\theta$ without loss of generality (If not, drop the $k$-th agent). In the main text, we call this expression an outcome-independent combination since $c_k(\theta)$ is not dependent on $y$. Instead, this is not the combination of the original agent densities $h_k$, but of the tilted densities $h_k^{\ast}$. To compute the integral in $c_k(\theta)$, we need to specify weight $w_k(u;\theta)$ and agent density $h_k(u)$. Here, we assume the bias-adjusted synthesis with the Gaussian weight, and that each agent density is Gaussian.

\begin{prp}
\label{prp:gauss-marginal}
Suppose $h_k(x)=\mathcal{N}(x\mid \mu_{h_k},\sigma_{h_k}^2)$ for $k=1,\ldots,K$, and consider the synthesis function
\begin{equation*}
    \alpha(y\mid x,\theta,\beta,\tau)
    = w_0(x;\theta) h_0(y)
    + \sum_{k=1}^K w_k(x_k;\theta)\,
      \mathcal{N}(y\mid x_k-\beta_k,\tau_k^2),
\end{equation*}
with Gaussian weights (\ref{eq:gauss-weight}). Then
\begin{equation}
    p(y\mid\mathcal{H},\theta,\beta,\tau)
    = \left\{ 1 - \sum_{k=1}^K q_k A_k \right\} h_0(y)
    + \sum_{k=1}^K q_k A_k\,
      \mathcal{N}\!\left( y \;\middle|\; a_k-\beta_k,\ \tau_k^2+r_k \right),
    \label{eq:gauss-marginal-closed}
\end{equation}
where
\begin{equation*}
    A_k
    = \sqrt{\frac{\sigma_k^2}{\sigma_k^2+\sigma_{h_k}^2}}
      \exp\left\{-\frac{(\mu_k-\mu_{h_k})^2}{2(\sigma_k^2+\sigma_{h_k}^2)}\right\},
    \qquad
    r_k = \left( \sigma_{h_k}^{-2}+\sigma_k^{-2} \right)^{-1},
\end{equation*}
\begin{equation*}
    a_k
    = r_k\left( \frac{\mu_{h_k}}{\sigma_{h_k}^2}
        + \frac{\mu_k}{\sigma_k^2} \right).
\end{equation*}
\end{prp}

The proof is a direct Gaussian integration. In the notation of (\ref{eq:gauss-marginal-generic}), the total mixture proportion of agent $k$ is $c_k=q_kA_k$. The corresponding tilted component is the normal density in (\ref{eq:gauss-marginal-closed}). As $\sigma_k^2\to\infty$ for all $k$, we have $A_k\to1$, $r_k\to\sigma_{h_k}^2$, and $a_k\to\mu_{h_k}$. In this case, the baseline component vanishes and the predictive density converges to the constant-weight calibrated linear pool with weights $q_{1:K}$. If $\sigma_k^2\to0$, then $A_k\to0$, so agent $k$ disappears from the predictive density and its share is reallocated to the baseline. These limiting behaviors are commented on in Section~\ref{sec:gaussian-asymptotics}.

The marginal density under the Dirac synthesis can be derived similarly. In fact, the expression above covers such cases by setting $\beta_k=0$ and $\tau_k=0$. These results are needed in evaluating the asymptotic predictive density in Section~S3.

\section{Posterior concentration for BPS with parameterized synthesis weights}
\label{supp:consistency}

This section proves the posterior concentration result used in Proposition~\ref{prp:bps-weight-concentration} of the main text. The argument is stated for a generic finite-dimensional parameterization of the synthesized predictive density. It therefore applies to the Gaussian and softmax weights, as well as other BPS synthesis functions, whenever the induced predictive density is proper and satisfies the regularity conditions below. The proof is a standard likelihood-ratio argument for posterior concentration under possible misspecification, specialized to the likelihood induced by the synthesized BPS predictive density \citep[Definition~6.1, Proposition~6.2, Theorems~6.16--6.17]{ghosalvandervaart2017fundamentals}.

\subsection{Setup and regularity conditions}

Let \(\mathcal{H}=\{h_1,\ldots,h_K\}\) be the fixed collection of agent forecast densities, let \(x=(x_1,\ldots,x_K)^\top\) denote latent agent states, and write $h(x)=\prod_{k=1}^K h_k(x_k)$.
Let $(\Theta,d)$ be a metric parameter space, equipped with its Borel $\sigma$-field $\mathcal{B}(\Theta)$, and let $\Pi$ be a prior probability measure on $\mathcal{B}(\Theta)$.

A parameter \(\theta\in\Theta\) enters through a synthesis function \(\alpha(y\mid x,\theta)\). The induced synthesized predictive density is
\begin{equation*}
    m_\theta(y)
    = p(y\mid\mathcal{H},\theta)
    = \int \alpha(y\mid x,\theta)h(x)\,dx,
\end{equation*}
where, for notational clarity, we suppressed the dependence on \(\mathcal{H}\).

Let \(Y_1,\ldots,Y_n\) be independently generated from some distribution \(P^\ast\). 
For $B\in\mathcal{B}(\Theta)$, define the synthesis likelihood and posterior by
\begin{equation*}
    L_n(\theta)=\prod_{i=1}^n m_\theta(Y_i),
    \qquad {\rm and} \qquad 
    \Pi_n(B\mid Y_{1:n},\mathcal{H})
    =
    \frac{\int_B L_n(\theta)\,\Pi(d\theta)}
         {\int_\Theta L_n(\theta)\,\Pi(d\theta)} .
\end{equation*}
Denote the expected log score by 
\begin{equation*}
    M(\theta)=\mathbb{E}_{P^\star}\!\left[\log m_\theta(Y)\right].
\end{equation*}
Suppose that \(P^\ast\) has density \(p^\ast\) with respect to the same dominating measure. Then, the expected log score is written as 
\begin{equation*}
    M(\theta)
    =
    \mathbb{E}_{P^\ast}\log p^\ast(Y)
    -
    \mathrm{KL}(p^\ast,m_\theta),
\end{equation*}
which is well-defined whenever the relevant expectations are finite. 
Hence maximizing the expected log score is equivalent to minimizing the Kullback--Leibler divergence from $p^\star$ to $m_\theta$. Next, define
\begin{equation*}
    M^\star:=\sup_{\theta\in\Theta}M(\theta),
    \qquad {\rm and} \qquad 
    \Theta^\star:=\{\theta\in\Theta:M(\theta)=M^\star\}.
\end{equation*}
For every $\theta^\star\in\Theta^\star$, the predictive density $m_{\theta^\star}$ is a KL projection of $p^\star$ onto the family $\{m_\theta:\theta\in\Theta\}$.
Under misspecification, \(\Theta^\star\) is the pseudo-true set \citep[Theorem~8.36]{ghosalvandervaart2017fundamentals}. Finally, for $\varepsilon>0$ and $\eta>0$, define
\[
    A_\varepsilon
    :=
    \{\theta\in\Theta:
      d(\theta,\Theta^\star)\geq\varepsilon\},
    \qquad {\rm and} \qquad 
    B_\eta
    :=
    \{\theta\in\Theta:
      M(\theta)\geq M^\star-\eta\}.
\]
\begin{asm}[Uniform log-score approximation and separation]
\label{asm:general-bps-consistency}
\hspace{0pt}
\begin{enumerate}
\item[(i)] For every \(\theta\in\Theta\), \(m_\theta\) is a proper density, \(m_\theta(Y)>0\) \(P^\ast\)-almost surely, and \(M(\theta)\) is finite.
\item[(ii)] The log score satisfies the uniform law of large numbers
\begin{equation*}
    \sup_{\theta\in\Theta}
    \left|
    \frac{1}{n}\sum_{i=1}^n\log m_\theta(Y_i)-M(\theta)
    \right|
    \xrightarrow{\ P^\ast\ }0 .
\end{equation*}
where $\xrightarrow{P^\star}$ denotes convergence in probability under $P^\star$.
\item[(iii)] \(M^\star\) is finite and \(\Theta^\star\) is nonempty.
\item[(iv)] For every \(\varepsilon>0\), there exists \(\delta_\varepsilon>0\) such that
\begin{equation*}
    \sup_{\theta\in A_\varepsilon}M(\theta)
    \leq
    M^\star-\delta_\varepsilon.
\end{equation*}
\item[(v)] For every \(\eta>0\),
\begin{equation*}
    \Pi(B_\eta)>0.
\end{equation*}
\end{enumerate}
\end{asm}

\begin{thm}[Posterior concentration for parameterized BPS weights]
\label{thm:general-bps-consistency}
Suppose Assumption~\ref{asm:general-bps-consistency} holds. Then, for every \(\varepsilon>0\),
\begin{equation*}
    \Pi_n\{\theta:d(\theta,\Theta^\star)\ge\varepsilon\mid Y_{1:n},\mathcal{H}\}
    \xrightarrow{\ P^\ast\ }0 .
\end{equation*}
If \(\Theta^\star=\{\theta^\star\}\), then
\begin{equation*}
    \Pi_n\{\theta:d(\theta,\theta^\star)\ge\varepsilon\mid Y_{1:n},\mathcal{H}\}
    \xrightarrow{\ P^\ast\ }0,
    \qquad \varepsilon>0 .
\end{equation*}
\end{thm}

\medskip
\noindent
\textit{Proof.}
Fix $\varepsilon>0$.
By Assumption~\ref{asm:general-bps-consistency}(iv), there exists
$\delta_\varepsilon>0$ such that
\[
    \sup_{\theta\in A_\varepsilon}M(\theta)
    \leq M^\star-\delta_\varepsilon.
\]
Choose $0<\eta<\delta_\varepsilon/3$ and define
\begin{equation*}
    \mathcal{E}_n(\eta)
    :=
    \left\{
    \sup_{\theta\in\Theta}
    \left|
    n^{-1}\sum_{i=1}^n\log m_\theta(Y_i)-M(\theta)
    \right|
    \le \eta
    \right\}.
\end{equation*}
By Assumption~\ref{asm:general-bps-consistency}(ii), \(P^\ast\{\mathcal{E}_n(\eta)\}\to1\). 

On $\mathcal E_n(\eta)$, for every $\theta\in A_\varepsilon$,
\[
    \frac{1}{n}\log L_n(\theta)
    \leq
    M(\theta)+\eta
    \leq
    M^\star-\delta_\varepsilon+\eta,
\]
and hence
\[
    \int_{A_\varepsilon}L_n(\theta)\,\Pi(d\theta)
    \leq
    \exp\{n(M^\star-\delta_\varepsilon+\eta)\}.
\]
Similarly, for every $\theta\in B_\eta$, we have 
\[
    \frac{1}{n}\log L_n(\theta)
    \geq
    M(\theta)-\eta
    \geq
    M^\star-2\eta.
\]
Therefore,
\[
    \int_\Theta L_n(\theta)\,\Pi(d\theta)
    \geq
    \int_{B_\eta}L_n(\theta)\,\Pi(d\theta)
    \geq
    \Pi(B_\eta)\exp\{n(M^\star-2\eta)\}.
\]
Note that Assumption~\ref{asm:general-bps-consistency}(v) ensures that
$\Pi(B_\eta)>0$, hence the right hand side above is strictly positive. 
Consequently, on $\mathcal E_n(\eta)$,
\[
    \Pi_n(A_\varepsilon\mid Y_{1:n},\mathcal H)
    \leq
    \Pi(B_\eta)^{-1}
    \exp\{-n(\delta_\varepsilon-3\eta)\}.
\]
Since $\eta<\delta_\varepsilon/3$, the right-hand side converges to zero
as $n\to\infty$. Together with
$P^\star\{\mathcal E_n(\eta)\}\to1$, this proves the result.
\hfill\(\square\)
\medskip

\begin{cor}[Correct specification and parameter consistency]
\label{cor:correctly-specified}
Suppose Assumption~\ref{asm:general-bps-consistency} holds. 
Suppose further that the model is correctly specified and identifiable:
there exists $\theta_0\in\Theta$ such that
\[
    m_{\theta_0}=p^\star
    \quad\text{a.e.},
\]
and
\[
    m_\theta=p^\star\ \text{a.e.}
    \quad\Longrightarrow\quad
    \theta=\theta_0.
\]
Then $\Theta^\star=\{\theta_0\}$, and the posterior is consistent at
$\theta_0$.
\end{cor}

\medskip
\noindent
\textit{Proof.}
For any \(\theta\in\Theta\),
\begin{equation*}
    M(\theta)
    =
    \mathbb{E}_{P^\ast}\log p^\ast(Y)-\mathrm{KL}(p^\ast,m_\theta)
    \le
    \mathbb{E}_{P^\ast}\log p^\ast(Y),
\end{equation*}
Equality holds if and only if $m_\theta=p^\star$ almost everywhere.
By identifiability, this occurs only at $\theta=\theta_0$.
Hence $\Theta^\star=\{\theta_0\}$, and Theorem~\ref{thm:general-bps-consistency} applies.
\hfill\(\square\)
\medskip

\begin{rem}[Sufficient primitive conditions]
\label{rem:primitive-route}
A conventional sufficient set of conditions for Assumption~\ref{asm:general-bps-consistency}(ii)--(v) in a finite-dimensional model is as follows.
Suppose that $\Theta$ is compact, that $\theta\mapsto\log m_\theta(y)$ is continuous for
$P^\star$-almost every $y$, and that there exists an integrable envelope $G$ satisfying
\[
    \sup_{\theta\in\Theta}
    |\log m_\theta(y)|
    \leq G(y),
    \qquad {\rm and} \qquad 
    \mathbb{E}_{P^\star}[G(Y)]<\infty.
\]
Then a uniform law of large numbers gives Assumption~\ref{asm:general-bps-consistency}(ii), while continuity of $M$ and compactness of $\Theta$ imply existence of a maximizer and the separation condition in Assumption~\ref{asm:general-bps-consistency}(iii)--(iv).
If, in addition, $\Pi$ has full support on $\Theta$, then Assumption~\ref{asm:general-bps-consistency}(v) follows from continuity of $M$.
\end{rem}

\begin{rem}[Gaussian and softmax weights]
\label{rem:gaussian-softmax-consistency}
For Gaussian weights, it is often convenient to compactify the width parameter by \(u_k=(1+\sigma_k^2)^{-1}\in[0,1]\). The endpoint \(u_k=0\) represents the limiting predictive density obtained as \(\sigma_k^2\to\infty\), and \(u_k=1\) represents the limiting predictive density obtained as \(\sigma_k^2\to0\). Thus Theorem~\ref{thm:general-bps-consistency} applies to the compactified Gaussian-weight family after the induced predictive density is continuously extended to these endpoints. Posterior concentration at a boundary means concentration in neighborhoods of the corresponding endpoint in the \(u_k\)-scale; it does not require the prior to put positive mass exactly at \(\sigma_k^2=0\) or \(\sigma_k^2=\infty\). 

\end{rem}

\section{KL projection in the two-part normal example}
\label{supp:klproj}

This section gives the details on the numerical evaluation of the asymptotic predictive density reported in
Section~\ref{sec:gaussian-illustration} of the main text.
As in Section~\ref{sec:gaussian-illustration}, the analysis here focuses on the Dirac synthesis with the Gaussian weights when the data generating process is the mixture of two truncated normal distributions. 

\subsection{Setting and Dirac Gaussian-weight BPS marginal density}

Remember that the density of the data generating process, or the true density, is the two-part normal distribution, 
\begin{equation*}
    p^\ast(y)
    =
    \frac{1}{3}\,
    \mathcal{TN}_{(-\infty,0]}(y\mid0,1^2)
    +
    \frac{2}{3}\,
    \mathcal{TN}_{[0,\infty)}(y\mid0,2^2),
\end{equation*}
where $\mathcal{TN}_A$ denotes the truncated normal density whose support is
$A\subset\mathbb{R}$. 
To estimate this density, we have two agent densities: 
\[
    h_1(x)=\mathcal{N}(x\mid0,1^2)
    \qquad {\rm and} \qquad
    h_2(x)=\mathcal{N}(x\mid0,2^2).
\]
The components of the data generating process are not identical to the two agents: the former is the pair of two truncated normal distributions, while the latter is not truncated. Hence, the mixture BPS models is misspecified, being unable to cover the true density. Still, with an appropriate truncation of the two agent densities, the mixture BPS would provide an accurate estimate of the predictive density, by using agent 1 on the left (negative) half-line and agent 2 on the right. In other words, this example favors the outcome-dependent way of combining the given agent densities.

In this example, we consider two baseline models. Let
\[
    {\rm base} \in\{\mathrm{diff},\mathrm{match}\},
\]
which is the index for the baseline specification. For each choice of the baseline, define
\[
    h_0^{(\mathrm{diff})}(y)
    =
    \mathcal{N}(y\mid0,10^2)
    \qquad {\rm and} \qquad
    h_0^{(\mathrm{match})}(y)
    =
    \mathcal{N}(y\mid m_B,s_B^2),
\]
where $m_B=\sqrt{2/\pi}$ and $s_B^2=3-2/\pi$. These values of $m_B$ and $s_B^2$ are the mean and variance of $p^\ast$.

Since the agent densities are Gaussian, we can compute the marginal predictive density by using the result of Proposition~\ref{prp:gauss-marginal}. In this example, we substitute $(\mu_{h_1},\sigma_{h_1}^2)=(0,1)$ and $(\mu_{h_2},\sigma_{h_2}^2)=(0,4)$ in the expression of Proposition~\ref{prp:gauss-marginal} to obtain 
\begin{equation*}
    m_\theta^{({\rm base})}(y)
    =
    \left\{
        1-\sum_{k=1}^2 q_kA_k
    \right\}
    h_0^{({\rm base})}(y)
    +
    \sum_{k=1}^2
    q_kA_k
    \mathcal{N}\!\left(
        y\,\middle|\,a_k,r_k
    \right).
\end{equation*}
In this expression, the outcome-independent weights are 
\[
    c_k(\theta)=q_kA_k
    \qquad {\rm and} \qquad
    c_0(\theta)=1-\sum_{k=1}^2 c_k(\theta).
\]

As a special case, consider the density at a boundary of the parameter space, where $\sigma _1^2 \to \infty$ and $\sigma _2^2 \to \infty$. 
Then we have
$A_k\to1$,
$r_k\to\sigma_{h_k}^2$,
and
$a_k\to0$. 
Hence the weights become constant, $c_k(\theta) \to q_k$, and the baseline receives no weight, $c_0(\theta) \to 1- (q_1+q_2) = 0$. The marginal density reduces to, 
\begin{equation*}
    m_\theta^{(\infty)}(y)
    =
    q_1\mathcal{N}(y\mid0,1)
    +
    q_2\mathcal{N}(y\mid0,4).
\end{equation*}
Thus, in the case of $\sigma _1^2 \to \infty$ and $\sigma _2^2 \to \infty$, the combined density is simply the constant-weight linear pool of the original agent densities, without introducing the truncation. 

\subsection{Expected log score and numerical projection}

For each baseline, the population expected log score is given by
\begin{equation}
\begin{split}
    M_{\rm base}(\theta)
    &=
    \mathbb{E}_{P^\ast}
    \Big[
        \log m_\theta^{({\rm base})}(Y)
    \Big] \\
    &= \mathbb{E}_{P^\ast}\Big[ \log p^\ast(Y) \Big] - \mathrm{KL}\!\left(
        p^\ast,m_\theta^{({\rm base})}
    \right)
    \label{eq:supp-objective}
\end{split}
\end{equation}
Thus, maximizing $M_{\rm base}(\theta)$ gives the KL projection of $p^\ast$ onto the
Gaussian-weight BPS family. 
We assume the parameter space is compact and numerically compute such a maximizer, denoted by $\theta_{\rm base}^\star$. 

The integral with respect to $p^{\ast}$ in \eqref{eq:supp-objective} is numerically approximated by applying the midpoint rule after transforming each truncated-normal component to the unit interval \citep{DavisRabinowitz1984}. 
We first define $N_q=5000$ equally-spaced grids,
\[
    u_j=\frac{j-1/2}{N_q},
    \qquad
    j=1,\ldots,N_q.
\]
Then, for each component of $p^{\ast}$, we set the $N_q$ nodes as 
\[
    y_{1j}
    =
    \Phi^{-1}(u_j/2),
    \qquad {\rm and} \qquad
    y_{2j}
    =
    2\Phi^{-1}(1/2+u_j/2),
\]
for $j=1,\dots ,N_q$, where $\Phi$ is the cumulative disrtibution function of the standard normal distribution. 
The objective function is approximated as 
\begin{equation*}
    M_{\rm base}(\theta)
    \simeq
    \frac{1}{3N_q}
    \sum_{j=1}^{N_q}
    \log m_\theta^{({\rm base})}(y_{1j})
    +
    \frac{2}{3N_q}
    \sum_{j=1}^{N_q}
    \log m_\theta^{({\rm base})}(y_{2j}).
\end{equation*}
We numerically maximize this approximate objective function over the five parameters in $\theta = (\mu_1, \mu_2, \log(\sigma_1^2), \log(\sigma_2^2), \operatorname{logit}(q_1))$. 
In doing so, we use the L-BFGS-B algorithm \citep{ByrdEtAl1995}, as implemented in the \texttt{optim} function in R (\texttt{method = "L-BFGS-B"}), allowing for 2500 iterations at most for each initialization.
Since the objective is not globally concave, we consider 60 different initial values, combining several structured starting values with additional values generated uniformly over the parameter space.
To evaluate the objective function at the boundary, we conduct the optimization in the following three cases separately: finite $\sigma_1^2<\infty$ and $\sigma_2^2<\infty$; $\sigma_1^2 =\infty$ and $\sigma_2^2<\infty$; and $\sigma_1^2 < \infty$ and $\sigma_2^2 =\infty$. 
Finally, the parameter space is explored by the bounded finite-width searches: 
\[
    10^{-6}\le\sigma_k^2\le10^8,
    \qquad
    10^{-8}\le q_1\le1-10^{-8}
    \qquad {\rm and} \qquad 
    |\mu_k|\le12.
\]
Note that the case of $\sigma_1^2 = \sigma_2^2 = \infty$ is available in the closed form as derived in the previous subsection.

Once we finalize the value of $\theta^\ast_{\mathrm{base}}$ using the fixed-grid approximation above, we re-evaluate the objective function with the parameter fixed at $\theta^\ast_{\mathrm{base}}$ using the R function \texttt{integrate} over $(-\infty,0]$ and $[0,\infty)$, with relative tolerance $10^{-8}$ and at most 500 subdivisions. The functional values computed here are reported in Table~1 of the main text. 

\subsection{Results and interpretation}

Under the diffuse baseline, the pseudo-true values of the window parameters and the baseline share are, 
\[
    \sigma_1^{2\star}=\infty
    \qquad
    \sigma_2^{2\star}=119.23,
    \qquad {\rm and} \qquad 
    c_0(\theta_{\mathrm{diff}}^\star)=0.0208.
\]
Hence the first agent receives a constant weight, while the second agent
retains outcome dependence slightly by having finite $\sigma_2^{2\star}$. 
However, the difference from the result under $\sigma_1^2=\sigma_2^2=\infty$ in the expected log score is
$9.53\times10^{-4}$ and small. 

Under the moment-matched baseline, both window parameters are chosen to be finite; 
\[
    \sigma_1^{2\star}=1.736,
    \qquad
    \sigma_2^{2\star}=6.391
    \qquad {\rm and} \qquad 
    c_0(\theta_{\mathrm{match}}^\star)=0.5765.
\]
The improvement over the boundary case in the expected log score is $0.12959$, so finite
outcome-dependent weights are clearly preferred for this baseline.

\begin{table}[htbp]
\centering
\caption{
KL-projection parameters under the Dirac synthesis.
Each row reports the pseudo-true quantities under the indicated baseline.
Here $q_2^\star=1-q_1^\star$.
A dash indicates a parameter that is irrelevant and unidentified at an
infinite-width boundary.
}
\label{tab:supp-klproj-finite}
\small
\setlength{\tabcolsep}{6pt}
\begin{tabular}{lcccccc}
\hline
Baseline
& $q_1^\star$
& $\mu_1^\star$
& $\mu_2^\star$
& $c_0^\star$
& $c_1^\star$
& $c_2^\star$ \\
\hline
diffuse
& 0.34646
& --
& 1.9794
& 0.020833
& 0.34646
& 0.63270 \\
matched
& 0.31448
& -0.46476
& 4.7358
& 0.57647
& 0.24081
& 0.18271 \\
\hline
\end{tabular}
\end{table}

The computed values of $\theta^{\ast}_{\rm base}$ other than $\sigma_k^{\ast}$ are summarized in Table~\ref{tab:supp-klproj-finite}. Using these values, one can compute the tilted component, $\mathcal{N}\!\left( y\,\middle|\,a_k,r_k \right)$, derived in Section~S3.1. 
For the diffuse baseline, the second component has $a_2=0.064$, so the effect of tilting is small. 
Under the moment-matched baseline,
\[
    a_1=-0.170,
    \qquad {\rm and} \qquad 
    a_2=1.82,
\]
so the first component is tilted slightly to the left and the second
component to the right.
These results confirm the interpretation in the main text:
the diffuse-baseline projection remains close to a constant-weight pool,
whereas the moment-matched baseline supports finite outcome-dependent
weights.
In neither case do the Gaussian weights recover the target hard-switching
pattern.

\section{MCMC algorithms for the BPS models}
\label{supp:mcmc}

\subsection{Gaussian weights}
\label{sec:supp-mcmc-gaussian}

We consider the bias-adjusted Gaussian-weight BPS with $K$ agents.
For observation $i$, the weight of agent $j$ is
\begin{align*}
  w_j(x_{ji})
  =
  q_j
  \exp\left\{
    -\frac{(x_{ji}-\mu_j)^2}{2\sigma_j^2}
  \right\},
  \qquad j=1,\ldots,K,
\end{align*}
and the baseline weight is
\begin{align*}
  w_0(x_i)
  =
  1-\sum_{j=1}^K w_j(x_{ji}).
\end{align*}
We sample
$q=(q_1,\ldots,q_K)$,
$\mu=(\mu_1,\ldots,\mu_K)$,
$\sigma^2=(\sigma_1^2,\ldots,\sigma_K^2)$,
the latent states
$x_i=(x_{1i},\ldots,x_{Ki})$,
the allocation indicators
$z_i\in\{0,1,\ldots,K\}$,
and the calibration parameters $(\beta_j,v_j)$.

Below, $h_j$ denotes the predictive distribution of agent $j$.
We assume that independent simulation from $h_j$ is available.
When $h_j$ admits a density, we use the same notation for that density.
The formulation also covers predictive distributions represented by Monte Carlo particles, with $h_j$ then interpreted as the corresponding discrete empirical distribution.
We define
\[
  I_j=\{i:z_i=j\},
  \qquad
  n_j=|I_j|,
  \qquad
  j=0,1,\ldots,K.
\]

\paragraph{Sampling scheme.}
Up to normalizing constants, the joint distribution is
\begin{align*}
  p(y,z,x,\Theta)
  &\propto
  \prod_{i=1}^n
  \left[
    \{w_0(x_i)h_0(y_i)\}^{\mathbb{I}(z_i=0)}
    \prod_{j=1}^K
    \{w_j(x_{ji})N(y_i\mid x_{ji}-\beta_j,v_j)\}^{\mathbb{I}(z_i=j)}
  \right]
  \prod_{i=1}^n\prod_{j=1}^K h_j(x_{ji})       \\
  &\quad\times
  p(q)p(\mu)p(\sigma^2)p(\beta)p(v).
\end{align*}
The sampler therefore combines Gibbs steps based on complete conditional distributions with Metropolis--Hastings steps that correct only for the non-conjugate components.

\paragraph{(1) Updating allocation indicators $z_i$.}
For each $i$, normalize
\begin{align*}
  P(z_i=0\mid\cdot)
  &\propto
  w_0(x_i)h_0(y_i),\\
  P(z_i=j\mid\cdot)
  &\propto
  w_j(x_{ji})N(y_i\mid x_{ji}-\beta_j,v_j),
  \qquad j=1,\ldots,K,
\end{align*}
and sample
\begin{align*}
  z_i\mid\cdot \sim \mathrm{Cat}(p_{i0},p_{i1},\ldots,p_{iK}).
\end{align*}

\paragraph{(2) Updating latent predictive states $x_{ji}$.}
The complete conditional distribution of $x_{ji}$ has three cases, depending on the current value of $z_i$.
First, if $z_i=j$, then
\begin{align*}
  p(x_{ji}\mid z_i=j,\cdot)
  &\propto
  h_j(x_{ji})
  w_j(x_{ji})
  N(y_i\mid x_{ji}-\beta_j,v_j).
\end{align*}
For a general agent density $h_j$, we use an independence
Metropolis--Hastings update with proposal
\begin{align*}
  x_{ji}^{\ast}\sim h_j.
\end{align*}
Since the proposal density cancels with the agent density in the target,
the acceptance probability is
\begin{align*}
  \alpha_{x,ji}
  =
  \min\left\{
  1,
  \frac{
    w_j(x_{ji}^{\ast})
    N(y_i\mid x_{ji}^{\ast}-\beta_j,v_j)
  }{
    w_j(x_{ji}^{(t)})
    N(y_i\mid x_{ji}^{(t)}-\beta_j,v_j)
  }
  \right\}.
\end{align*}

If $z_i=k\ne j$ with $k\in\{1,\ldots,K\}$, $x_{ji}$ does not enter
the likelihood or allocation contribution for observation $i$, and
hence it can be sampled directly from
\begin{align*}
  x_{ji}\mid z_i=k,\cdot\sim h_j.
\end{align*}

Finally, if $z_i=0$, the baseline weight $w_0(x_i)$ depends on
$x_{ji}$. We again use the independence proposal
\begin{align*}
  x_{ji}^{\ast}\sim h_j
\end{align*}
and accept with probability
\begin{align*}
  \alpha_{x,ji}
  =
  \min\left\{
  1,
  \frac{
    w_0(x_{ji}^{\ast},x_{-ji}^{(t)})
  }{
    w_0(x_{ji}^{(t)},x_{-ji}^{(t)})
  }
  \right\}.
\end{align*}
In this Metropolis--Hastings ratio, the proposal density and the agent prior cancel.

When the agent density is Gaussian,
\[
  h_j(x)
  =
  N(x\mid\mu_{h_j},\sigma_{h_j}^2),
\]
the update for $z_i=j$ is available directly as a Gibbs step:
\begin{align*}
  x_{ji}\mid z_i=j,\cdot
  &\sim N(m_{x,ji},s_{x,j}^2),
\end{align*}
where
\begin{align*}
  (s_{x,j}^2)^{-1}
  &=
  \sigma_{h_j}^{-2}
  +\sigma_j^{-2}
  +v_j^{-1},\\
  m_{x,ji}
  &=
  s_{x,j}^2
  \left\{
    \mu_{h_j}\sigma_{h_j}^{-2}
    +\mu_j\sigma_j^{-2}
    +(y_i+\beta_j)v_j^{-1}
  \right\}.
\end{align*}
We use this conjugate special case in the numerical experiments.

\paragraph{(3) Updating the maximum weights $q$.}
Let $q\sim\mathrm{Dir}(\alpha_1,\ldots,\alpha_K)$.
The complete conditional distribution is
\begin{align*}
p(q\mid\cdot)
\propto
\mathrm{Dir}
\left(
q\mid
\alpha_1+n_1,\ldots,\alpha_K+n_K
\right)
\prod_{i\in I_0}
w_0(x_i\mid q,\mu,\sigma^2).
\end{align*}
We use the conjugate component
\begin{align*}
  q^\ast
    \sim
    \mathrm{Dir}
    (\alpha_1+n_1,\ldots,\alpha_K+n_K).
\end{align*}
as an independence Metropolis--Hastings proposal and accept with probability
\begin{align*}
  \alpha_q
  =\min\left\{1,
  \frac{\prod_{i\in I_0}w_0(x_i\mid q^\ast,\mu,\sigma^2)}
       {\prod_{i\in I_0}w_0(x_i\mid q^{(t)},\mu,\sigma^2)}
  \right\}.
\end{align*}
If $I_0=\varnothing$, the proposed $q^\ast$ is accepted automatically.

\paragraph{(4) Updating the region centers $\mu_j$.}
Let $\mu_j\sim N(m_{\mu,j}^{(0)},s_{\mu,j}^{2(0)})$.
The complete conditional distribution is
\begin{align*}
  p(\mu_j\mid\cdot)
  \propto
  N(\mu_j\mid m_{\mu,j}^{(0)},s_{\mu,j}^{2(0)})
  \prod_{i\in I_j}
  \exp\left\{-\frac{(x_{ji}-\mu_j)^2}{2\sigma_j^2}\right\}
  \prod_{i\in I_0}w_0(x_i\mid\mu_j).
\end{align*}
From the conjugate component, we propose
\begin{align*}
  \mu_j^\ast\sim N(m_{\mu,j},s_{\mu,j}^2),
\end{align*}
where
\begin{align*}
  (s_{\mu,j}^2)^{-1}
  &=(s_{\mu,j}^{2(0)})^{-1}+n_j\sigma_j^{-2},\\
  m_{\mu,j}
  &=s_{\mu,j}^2
  \left\{
    m_{\mu,j}^{(0)}(s_{\mu,j}^{2(0)})^{-1}
    +\sigma_j^{-2}\sum_{i\in I_j}x_{ji}
  \right\}.
\end{align*}
The proposal is accepted with probability
\begin{align*}
  \alpha_{\mu,j}
  =\min\left\{1,
  \frac{\prod_{i\in I_0}w_0(x_i\mid \mu_j^\ast,\mu_{-j},q,\sigma^2)}
       {\prod_{i\in I_0}w_0(x_i\mid \mu_j^{(t)},\mu_{-j},q,\sigma^2)}
  \right\}.
\end{align*}

\paragraph{(5) Updating $\sigma_j^2$.}
Let $\sigma_j^2\sim\mathrm{IG}(a_{\sigma,j},b_{\sigma,j})$.
The complete conditional distribution is
\begin{align*}
  p(\sigma_j^2\mid\cdot)
  &\propto
  \mathrm{IG}(\sigma_j^2\mid a_{\sigma,j},b_{\sigma,j})
  \prod_{i\in I_j}
  \exp\left\{-\frac{(x_{ji}-\mu_j)^2}{2\sigma_j^2}\right\}
  \prod_{i\in I_0}w_0(x_i\mid\sigma_j^2).
\end{align*}
Because the Gaussian weight does not include a normalizing constant, the number of observations in $I_j$, namely $n_j$, does not increase the shape parameter of the inverse-gamma distribution.
Therefore, the conjugate component is
\begin{align*}
  g_j(\sigma_j^2)
  =
  \mathrm{IG}\left(
    a_{\sigma,j},
    b_{\sigma,j}
    +\frac12\sum_{i\in I_j}(x_{ji}-\mu_j)^2
  \right).
\end{align*}
We propose
\begin{align*}
  \sigma_j^{2\ast}\sim g_j(\sigma_j^2)
\end{align*}
from this distribution and accept using the remaining ratio associated with the baseline constraint:
\begin{align*}
  \alpha_{\sigma,j}
  =\min\left\{1,
  \frac{\prod_{i\in I_0}w_0(x_i\mid \sigma_j^{2\ast},\sigma_{-j}^2,q,\mu)}
       {\prod_{i\in I_0}w_0(x_i\mid \sigma_j^{2(t)},\sigma_{-j}^2,q,\mu)}
  \right\}.
\end{align*}
If $\sigma_j^2$ is fixed at some specified value in the model of interest (e.g., the fixed-width Gaussian specification in the numerical experiments), skip Step~(5).

\paragraph{(6) Updating calibration parameters $\beta_j$ and $v_j$.}
Let $\beta_j\sim N(m_{\beta,j}^{(0)},s_{\beta,j}^{2(0)})$.
Conditional on $v_j$,
\begin{align*}
  \beta_j\mid\cdot
  \sim N(m_{\beta,j},s_{\beta,j}^2),
\end{align*}
where
\begin{align*}
  (s_{\beta,j}^2)^{-1}
  &=(s_{\beta,j}^{2(0)})^{-1}+n_jv_j^{-1},\\
  m_{\beta,j}
  &=s_{\beta,j}^2
  \left\{
    m_{\beta,j}^{(0)}(s_{\beta,j}^{2(0)})^{-1}
    +v_j^{-1}\sum_{i\in I_j}(x_{ji}-y_i)
  \right\}.
\end{align*}
If $v_j\sim\mathrm{IG}(a_{v,j},b_{v,j})$, then
\begin{align*}
  v_j\mid\cdot
  \sim
  \mathrm{IG}\left(
    a_{v,j}+\frac{n_j}{2},
    b_{v,j}+\frac12\sum_{i\in I_j}\{y_i-(x_{ji}-\beta_j)\}^2
  \right).
\end{align*}
When $n_j=0$, these complete conditional distributions reduce to their corresponding priors.

\subsection{Softmax weight}
\label{sec:supp-mcmc-softmax}
\subsubsection{Linear scores}

In the linear Softmax/MNL BPS specification, the baseline is treated as the reference category and we set
\begin{align*}
  \eta_{0i}&=0,\\
  \eta_{ji}&=b_{ji}'\gamma_j,
  \qquad
  b_{ji}=(1,x_{ji})',
  \qquad j=1,\ldots,K.
\end{align*}
The softmax weights are
\begin{align*}
w_{ji}
=
\frac{\exp(\eta_{ji})}
{\sum_{k=0}^{K}\exp(\eta_{ki})},
\qquad j=0,\ldots,K.
\end{align*}

\paragraph{(1) Updating allocation indicators $z_i$.}
For each $i$, normalize
\begin{align*}
  P(z_i=0\mid\cdot)&\propto w_{0i}h_0(y_i),\\
  P(z_i=j\mid\cdot)&\propto w_{ji}N(y_i\mid x_{ji}-\beta_j,v_j),
  \qquad j=1,\ldots,K
\end{align*}
and draw $z_i$ from $\mathrm{Cat}(p_{i0},p_{i1}, \ldots, p_{iK})$.

\paragraph{(2) Introducing P\'olya--Gamma auxiliary variables.}
For the multinomial logit component, we use blockwise P\'olya--Gamma augmentation \citep{polson2013bayesian}.
Treating the linear predictors of all categories other than $j$ as an offset, define
\begin{align*}
  c_{ji}=\log\sum_{k\ne j}\exp(\eta_{ki}),
  \qquad
  \psi_{ji}=\eta_{ji}-c_{ji},
  \qquad
  \kappa_{ji}=\mathbb{I}(z_i=j)-\frac12.
\end{align*}
For the conditional logit likelihood associated with category $j$, introduce
\begin{align*}
  \omega_{ji}\mid\cdot \sim \mathrm{PG}(1,\psi_{ji}).
\end{align*}
Writing
\begin{align*}
  \tilde\eta_{ji}=c_{ji}+\frac{\kappa_{ji}}{\omega_{ji}},
\end{align*}
the conditional likelihood for $\eta_{ji}$ becomes proportional to the Gaussian likelihood
\begin{align*}
  \exp\left\{-\frac{\omega_{ji}}{2}
  (\eta_{ji}-\tilde\eta_{ji})^2\right\}.
\end{align*}
In implementation, for each block $j=1,\ldots,K$, $c_{ji}$ is computed by fixing the current values of the other categories' $\eta_{ki}$, and then $\omega_{ji}$ and $\gamma_j$ are updated.
Since the baseline is fixed at $\eta_{0i}=0$, no coefficients are sampled for the baseline category.

\paragraph{(3) Updating Softmax coefficients $\gamma_j$.}
Let $\gamma_j\sim N(m_{\gamma,j}^{(0)},V_{\gamma,j}^{(0)})$.
Let $B_j$ denote the design matrix whose $i$th row is $b_{ji}'$, and define
\[
  \Omega_j
  =
  \operatorname{diag}(\omega_{j1},\ldots,\omega_{jn}),
\]
and
\[
  \tilde\eta_j
  =
  (\tilde\eta_{j1},\ldots,\tilde\eta_{jn})'.
\]
Then
\begin{align*}
  \gamma_j\mid\cdot
  \sim N(m_{\gamma,j},V_{\gamma,j}),
  \qquad j=1,\ldots,K,
\end{align*}
where
\begin{align*}
V_{\gamma,j}
&=
\left\{
  (V_{\gamma,j}^{(0)})^{-1}
  +
  B_j'\Omega_j B_j
\right\}^{-1},\\
m_{\gamma,j}
&=
V_{\gamma,j}
\left\{
  (V_{\gamma,j}^{(0)})^{-1}m_{\gamma,j}^{(0)}
  +
  B_j'\Omega_j\tilde\eta_j
\right\}.
\end{align*}

\paragraph{(4) Updating latent predictive states $x_{ji}$.}
Under the linear softmax specification,
\[
  \eta_{ji}
  =
  \gamma_{0j}+\gamma_{1j}x_{ji}.
\]
Before updating the latent states for agent $j$, we recompute
\[
  c_{ji}
  =
  \log\sum_{k\ne j}\exp(\eta_{ki})
\]
using the current values of all other categories, and resample
\[
  \omega_{ji}
  \mid\cdot
  \sim
  \mathrm{PG}
  \left(
    1,\,
    \eta_{ji}-c_{ji}
  \right),
  \qquad i=1,\ldots,n.
\]
We then redefine
\[
  \widetilde\eta_{ji}
  =
  c_{ji}
  +
  \frac{\mathbb{I}(z_i=j)-1/2}{\omega_{ji}}.
\]
This blockwise resampling is required because updates of the other
categories change the offset $c_{ji}$.
Conditional on the P\'olya--Gamma auxiliary variables, the complete
conditional kernel is
\begin{align*}
  p(x_{ji}\mid\cdot)
  \propto{}&
  h_j(x_{ji})
  \left\{
    N(y_i\mid x_{ji}-\beta_j,v_j)
  \right\}^{\mathbb{I}(z_i=j)}
  \\
  &\times
  \exp\left\{
    -\frac{\omega_{ji}}{2}
    \bigl(
      \eta_{ji}(x_{ji})-\widetilde\eta_{ji}
    \bigr)^2
  \right\}.
\end{align*}

For a general $h_j$, we use the independence proposal
\[
  x_{ji}^{\ast}\sim h_j.
\]
The corresponding Metropolis--Hastings acceptance probability is
\begin{align*}
  \alpha_{x,ji}
  =
  \min\Bigg\{1,\,
  &\left[
  \frac{
    N(y_i\mid x_{ji}^{\ast}-\beta_j,v_j)
  }{
    N(y_i\mid x_{ji}^{(t)}-\beta_j,v_j)
  }
  \right]^{\mathbb{I}(z_i=j)}
  \\
  &\times
  \exp\left[
    -\frac{\omega_{ji}}{2}
    \left\{
      \bigl(
        \eta_{ji}(x_{ji}^{\ast})-\widetilde\eta_{ji}
      \bigr)^2
      -
      \bigl(
        \eta_{ji}(x_{ji}^{(t)})-\widetilde\eta_{ji}
      \bigr)^2
    \right\}
  \right]
  \Bigg\}.
\end{align*}

When
\[
  h_j(x)
  =
  N(x\mid\mu_{h_j},\sigma_{h_j}^2),
\]
the complete conditional is Gaussian and the
Metropolis--Hastings step can be replaced by the Gibbs update
\begin{align*}
  x_{ji}\mid\cdot
  &\sim
  N(m_{x,ji},s_{x,ji}^2),
\end{align*}
where
\begin{align*}
  (s_{x,ji}^2)^{-1}
  &=
  \sigma_{h_j}^{-2}
  +\mathbb{I}(z_i=j)v_j^{-1}
  +\omega_{ji}\gamma_{1j}^2,\\
  m_{x,ji}
  &=
  s_{x,ji}^2
  \left\{
    \mu_{h_j}\sigma_{h_j}^{-2}
    +\mathbb{I}(z_i=j)(y_i+\beta_j)v_j^{-1}
    +\omega_{ji}\gamma_{1j}
      (\widetilde\eta_{ji}-\gamma_{0j})
  \right\}.
\end{align*}
This conjugate update is used in the numerical experiments.

\paragraph{(5) Updating calibration parameters.}
Since the Softmax weights do not depend on $(\beta_j,v_j)$, $\beta_j$ and $v_j$ can be updated by Gibbs sampling from the same normal and inverse-gamma distributions as in Gaussian weight.

\subsubsection{Quadratic scores}

The quadratic softmax specification uses quadratic score functions to represent nonlinear switching between a central region and the tails:
\begin{align*}
  \eta_{0i}&=0,\\
  \eta_{ji}
  &=
  b_{ji}'\gamma_j,
  \qquad
  b_{ji}=(1,x_{ji},x_{ji}^2)',
  \qquad j=1,\ldots,K.
\end{align*}
The Softmax weights are, as in the linear softmax specification,
\begin{align*}
w_{ji}
=
\frac{\exp(\eta_{ji})}
{\sum_{k=0}^{K}\exp(\eta_{ki})}.
\end{align*}

\paragraph{(1) Updating allocation indicators $z_i$.}
As in the linear softmax specification, $z_i$ is updated from the categorical distribution obtained by normalizing
\begin{align*}
  P(z_i=0\mid\cdot)&\propto w_{0i}h_0(y_i),\\
  P(z_i=j\mid\cdot)&\propto w_{ji}N(y_i\mid x_{ji}-\beta_j,v_j),
  \qquad j=1,\ldots,K.
\end{align*}

\paragraph{(2) Updating latent predictive states $x_{ji}$.}
Under the quadratic softmax specification, $\eta_{ji}$ is a quadratic function of $x_{ji}$.
Conditional on the P\'olya--Gamma auxiliary variables, $(\eta_{ji}-\tilde\eta_{ji})^2$ is therefore a quartic polynomial in $x_{ji}$, and the complete conditional distribution of $x_{ji}$ is no longer normal.
We consequently update $x_{ji}$ using a Metropolis--Hastings step based on the original Softmax likelihood, without conditioning on the P\'olya--Gamma auxiliary variables.

For a general agent density $h_j$, we use the independence proposal
\begin{align*}
  x_{ji}^{\ast}\sim h_j.
\end{align*}
Since the agent density cancels in the Metropolis--Hastings ratio, the acceptance probability is
\begin{align*}
  \alpha_{x,ji}
  =
  \min\Bigg\{
  1,\,
  \frac{
    w_{z_i i}
    (x_{ji}^{\ast},x_{-ji}^{(t)},\gamma)
  }{
    w_{z_i i}
    (x_{ji}^{(t)},x_{-ji}^{(t)},\gamma)
  }
  \left[
  \frac{
    N(y_i\mid x_{ji}^{\ast}-\beta_j,v_j)
  }{
    N(y_i\mid x_{ji}^{(t)}-\beta_j,v_j)
  }
  \right]^{\mathbb{I}(z_i=j)}
  \Bigg\}.
\end{align*}
Here $w_{z_i i}$ denotes the softmax probability of the currently allocated component, including the baseline when $z_i=0$.

When the agent density is Gaussian, a more efficient proposal is available. 
If $z_i=j$, we use the proposal proportional to the product of the agent density and calibration likelihood,
\[
  h_j(x_{ji})
  N(y_i\mid x_{ji}-\beta_j,v_j),
\]
which is
\[
  N(m_{ji}^{\ast},s_{ji}^{2\ast})
\]
with
\begin{align*}
  (s_{ji}^{2\ast})^{-1}
  &=
  \sigma_{h_j}^{-2}+v_j^{-1},\\
  m_{ji}^{\ast}
  &=
  s_{ji}^{2\ast}
  \left\{
    \mu_{h_j}\sigma_{h_j}^{-2}
    +(y_i+\beta_j)v_j^{-1}
  \right\}.
\end{align*}
If $z_i\ne j$, we simply draw $x_{ji}^{\ast}\sim h_j$.
For this Gaussian proposal, the agent density and calibration likelihood cancel, and the acceptance probability reduces to
\begin{align*}
  \alpha_{x,ji}
  =
  \min\left\{
  1,
  \frac{
    w_{z_i i}(x_{ji}^{\ast},x_{-ji}^{(t)},\gamma)
  }{
    w_{z_i i}(x_{ji}^{(t)},x_{-ji}^{(t)},\gamma)
  }
  \right\}.
\end{align*}
This Gaussian special case is used in the numerical experiments.

\paragraph{(3) Updating P\'olya--Gamma auxiliary variables and Softmax coefficients $\gamma_j$.}
After updating $x$, we update the coefficient blocks sequentially for $j=1,\ldots,K$. For each $j$, we first recompute $c_{ji}$ using the current values of all other categories, then sample the P\'olya--Gamma variables and update $\gamma_j$.
\begin{align*}
  c_{ji}=\log\sum_{k\ne j}\exp(\eta_{ki}),
  \qquad
  \omega_{ji}\mid\cdot\sim \mathrm{PG}(1,\eta_{ji}-c_{ji}),
  \qquad
  \tilde\eta_{ji}=c_{ji}+\frac{\mathbb{I}(z_i=j)-1/2}{\omega_{ji}}.
\end{align*}
Using $b_{ji}=(1,x_{ji},x_{ji}^2)'$, the complete conditional distribution of $\gamma_j$ is the same multivariate normal distribution as in the linear softmax specification. Let $\Omega_j = \operatorname{diag}(\omega_{j1},\ldots,\omega_{jn})$ as in the linear softmax specification. Then
\begin{align*}
  \gamma_j\mid\cdot\sim N(m_{\gamma,j},V_{\gamma,j}),
  \qquad j=1,\ldots,K,
\end{align*}
with
\begin{align*}
  V_{\gamma,j}
  &=\left\{(V_{\gamma,j}^{(0)})^{-1}+B_j'\Omega_jB_j\right\}^{-1},\\
  m_{\gamma,j}
  &=V_{\gamma,j}
  \left\{(V_{\gamma,j}^{(0)})^{-1}m_{\gamma,j}^{(0)}
  +B_j'\Omega_j\tilde\eta_j\right\}.
\end{align*}

\paragraph{(4) Updating calibration parameters.}
As in the Gaussian-weight and linear softmax specifications, $\beta_j$ and $v_j$ are updated from
\begin{align*}
  \beta_j\mid\cdot&\sim N(m_{\beta,j},s_{\beta,j}^2),\\
  v_j\mid\cdot&\sim
  \mathrm{IG}\left(
    a_{v,j}+\frac{n_j}{2},
    b_{v,j}+\frac12\sum_{i\in I_j}\{y_i-(x_{ji}-\beta_j)\}^2
  \right).
\end{align*}

\section{Additional details for the numerical experiments}
\label{supp:num}

\subsection{Common BPS specification}
\label{sec:supp-num-common}

The baseline density is fixed at $h_0(y)=\mathcal{N}(y\mid0,10^2)$.
For Gaussian weights, $w_0(x)=1-\sum_j w_j(x_j)$, whereas under the softmax specifications the baseline is the reference category with $\eta_{0i}=0$.

Each agent $j=1,2$ is associated with a latent state $x_j\sim h_j$.
Conditional on allocation to agent $j$, the calibrated observation model is
\[
  y\mid z=j,x_j,\beta_j,v_j
  \sim N(x_j-\beta_j,v_j).
\]
Here $\beta_j$ and $v_j$ are location and variance calibration parameters, respectively.
The corresponding priors are given in Section~\ref{sec:supp-num-priors}.

\subsection{Competing Methods}

\subsubsection{Equal-weight pool}

The equal-weight benchmark is
$p(y)=\{h_0(y)+h_1(y)+h_2(y)\}/3$ and requires no estimation.

\subsubsection{Linear pool}

We consider the linear pool
\begin{align*}
  p(y;w) = w_0 h_0(y) + w_1 h_1(y) + w_2 h_2(y),
  \qquad
  w_j\ge 0,
  \quad
  \sum_{j=0}^2 w_j = 1.
\end{align*}
The weight vector $w$ is estimated by maximizing the training-sample average log score
$n^{-1}\sum_{i=1}^n\log p(y_i;w)$
over the simplex.

\subsubsection{GDFC}

GDFC is an implementation of the generalized density forecast combination (GDFC) approach of \citet{kapetanios2015generalised}, adapted to the comparisons in this paper.
The numerical experiments in \citet{kapetanios2015generalised} combine two component densities using piecewise weight functions based on indicator bases.
For a fair comparison with the BPS methods, however, we use the same three densities available to the BPS specifications, $(h_0,h_1,h_2)$, as component densities in the combination.

\paragraph{Model.}
Let $q_j(y)=h_j(y)$ for $j=0,1,2$, and let
$-\infty=r_0<r_1<\cdots<r_{p-1}<r_p=\infty$
define the $p$ regions.
The GDFC density used here is
\begin{align}
  p_{\mathrm{GDFC}}(y\mid \nu,\mathcal S)
  =
  \sum_{j=0}^{2}\sum_{s=1}^{p}
  \tilde\nu_{js}\, q_j(y)\,\mathbb{I}(r_{s-1}\le y < r_s),
  \label{eq:gdfc_density}
\end{align}
where $\mathcal S=(p,r_1,\ldots,r_{p-1})$ denotes the region-partition specification and $\tilde\nu_{js}$ are normalized combination coefficients.

For each component density, define
\begin{align*}
  \kappa_{js}
  =
  \int_{r_{s-1}}^{r_s} q_j(u)\,du,
  \qquad j=0,1,2,\quad s=1,\ldots,p.
\end{align*}
For unnormalized weights $\nu_{js}>0$, set
\begin{align}
  \tilde\nu_{js}
  =
  \frac{\nu_{js}}{\sum_{\ell=0}^{2}\sum_{m=1}^{p}\nu_{\ell m}\kappa_{\ell m}}.
  \label{eq:gdfc_normweight}
\end{align}
By construction,
$\int p_{\mathrm{GDFC}}(y\mid\nu,\mathcal S)\,dy=1$,
so \eqref{eq:gdfc_density} is a proper density.

In the implementation, we reparameterize the positive coefficients as $\nu_{js}=\exp(\alpha_{js})$.
Since \eqref{eq:gdfc_normweight} is invariant to a common shift in the $\alpha_{js}$, we fix $\alpha_{01}=0$ for identification and numerical stability and optimize over the remaining parameters.

\paragraph{Interpretation of the weight functions.}
The corresponding region-wise coefficient is
\[
  w_j(y)
  =
  \sum_{s=1}^p
  \tilde\nu_{js}\mathbb I(r_{s-1}\le y<r_s),
  \qquad j=0,1,2,
\]
so that
$p_{\mathrm{GDFC}}(y\mid\nu,\mathcal S)
 =\sum_{j=0}^2w_j(y)q_j(y)$.
For visualization and weight-recovery diagnostics, we additionally define the pointwise-normalized GDFC coefficients
\[
  \bar w_j^{\mathrm{GDFC}}(y)
  =
  \frac{\tilde\nu_{j,s(y)}}
       {\sum_{\ell=0}^{2}\tilde\nu_{\ell,s(y)}},
  \qquad j=0,1,2,
\]
where $s(y)$ denotes the region containing $y$.
These normalized coefficients sum to one at each $y$ and are therefore more directly comparable with the BPS mixture weights.
The predictive density itself continues to use the globally normalized coefficients $\tilde\nu_{js}$ in \eqref{eq:gdfc_density}.

\paragraph{Number of regions and threshold grid.}
The numerical experiments of \citet{kapetanios2015generalised} use indicator bases and search over candidates on a fixed threshold grid.
We use the same candidate set for both DGPs,
\[
  p\in\{2,3,4\},
  \qquad
  r\in\{-3,-2.5,\ldots,2.5,3\},
\]
so that GDFC is not given a DGP-specific threshold grid.
The choices $p=2,3,4$ allow one, two, and three interior thresholds, respectively.
Thus, this candidate set can contain both the one-boundary structure of DGP-B and the two-boundary structure of DGP-A, while the actual boundary is selected by validation log score.
Restricting the grid to increments of 0.5 and excluding $p > 4$ avoids giving GDFC an excessive advantage from a very fine grid search and also reduces computation time.

\paragraph{Specification selection.}
\citet{kapetanios2015generalised} select $p$ using a pseudo-out-of-sample log score.
Following this idea, we use ordered validation rather than a random split, preserving the temporal ordering of the training sample.

For each candidate specification $\mathcal S$, we compute the average one-step-ahead log score over the validation block using the GDFC model estimated only from data prior to the start of each validation time point:
\begin{align*}
  \ell_{\mathrm{val}}(\mathcal S)
  =
  \frac{1}{|\mathcal I_{\mathrm{val}}|}
  \sum_{t\in \mathcal I_{\mathrm{val}}}
  \log p_{\mathrm{GDFC}}\bigl(y_t\mid \hat\alpha_{t-1}(\mathcal S),\mathcal S\bigr).
\end{align*}
Here $\hat\alpha_{t-1}(\mathcal S)$ denotes the estimator under specification $\mathcal S$ based on $(y_1,\ldots,y_{t-1})$.
The selected specification is
\begin{align*}
  \widehat{\mathcal S}
  =
  \operatorname*{arg\,max}_{\mathcal S} \ell_{\mathrm{val}}(\mathcal S).
\end{align*}

\paragraph{Estimation.}
For a fixed specification $\mathcal S$, the parameter vector $\alpha$ governing the coefficients, and hence the weights $\nu$, is estimated by maximizing the average log score on the training sample:
\begin{align*}
  \ell(\alpha;\mathcal S)
  =
  \frac{1}{n}
  \sum_{i=1}^{n}
  \log p_{\mathrm{GDFC}}(y_i\mid \alpha,\mathcal S).
\end{align*}
In numerical optimization, we solve
\begin{align*}
  \hat\alpha(\mathcal S)
  =
  \operatorname*{arg\,max}_{\alpha}
  \left\{
    \ell(\alpha;\mathcal S)
    - \lambda \sum \alpha_{js}^{2}
  \right\}.
\end{align*}
For numerical stability, we set the ridge penalty to $\lambda=10^{-4}$.

After $\widehat{\mathcal S}$ is selected, $\alpha$ is re-estimated using the full training sample.
The resulting predictive density is $p_{\mathrm{GDFC}}(y\mid\hat\alpha,\widehat{\mathcal S})$.
Auxiliary predictive samples from this density are used to evaluate sample-based criteria such as CRPS and interval coverage.

\subsection{Priors and Hyperparameters}
\label{sec:supp-num-priors}

The default prior specifications are as follows.

\paragraph{Gaussian-weight.}
We use
\[
  q\sim\mathrm{Dir}(2,2),
  \qquad
  \mu_j\sim N(0,4),
  \qquad j=1,2.
\]
For the Gaussian-weight specification with estimated widths,
\[
  \sigma_j^2\sim\mathrm{IG}(4,30),
  \qquad
  E[\sigma_j^2]=10,
  \qquad j=1,2,
\]
under the inverse-gamma shape--scale convention used here.
The chains are initialized at $\sigma_1^2=\sigma_2^2=10$, with $10^{-4}<\sigma_j^2<10^4$ imposed as numerical bounds.
The fixed-width specification instead sets $\sigma_1^2=\sigma_2^2=1$.

\paragraph{Linear softmax weight.}
For $j=1,2$,
\[
  \gamma_j
  =
  (\gamma_{0j},\gamma_{1j})'
  \sim N(0,10I_2).
\]

\paragraph{Quadratic softmax weight.}
For $j=1,2$,
\[
  \gamma_j
  =
  (\gamma_{0j},\gamma_{1j},\gamma_{2j})'
  \sim N(0,10I_3).
\]

\paragraph{Calibration prior.}
For $j=1,2$, we use
$\beta_j\sim N(0,1)$ and $v_j\sim\mathrm{IG}(5,1)$,
so that $E[v_j]=0.25$.
The same prior is used for both agents.

\subsection{Evaluation and figure construction}
\label{sec:supp-num-evaluation}

\paragraph{Predictive-density evaluation.}
Let
\[
    \psi=(\theta,\beta,\tau^2)
\]
collect the weight parameters and calibration parameters.
For a fixed value of $\psi$, the predictive density under the mixture BPS specification is
\begin{equation}
\begin{split}
p(y\mid\mathcal H,\psi)
={}&
\int
\Bigg[
    w_0(x;\theta)h_0(y)
    +
    \sum_{k=1}^K
    w_k(x;\theta)
    \mathcal N
    \left(
        y\mid x_k-\beta_k,\tau_k^2
    \right)
\Bigg]
\prod_{k=1}^K h_k(x_k)\,dx ,
\end{split}
\label{eq:supp-conditional-predictive}
\end{equation}
where $x=(x_1,\ldots,x_K)'$ denotes the latent agent states.
The posterior predictive density is obtained by averaging (\ref{eq:supp-conditional-predictive}) over the posterior distribution of $\psi$,
\begin{equation}
p(y\mid\mathcal H,\mathcal D)
=
\int
p(y\mid\mathcal H,\psi)
\pi(\psi\mid\mathcal D)\,d\psi .
\label{eq:supp-posterior-predictive}
\end{equation}
Given retained posterior draws
$\psi^{(1)},\ldots,\psi^{(S)}$, we approximate
(\ref{eq:supp-posterior-predictive}) by
\begin{equation}
\widehat p(y\mid\mathcal H,\mathcal D)
=
\frac{1}{S}
\sum_{s=1}^S
p(y\mid\mathcal H,\psi^{(s)}).
\label{eq:supp-posterior-predictive-mc}
\end{equation}

For the equal-weight pool and linear pool, the predictive density is
available directly as a finite mixture. For GDFC, the fitted predictive
density is evaluated directly from its piecewise representation.

\paragraph{Gaussian-weight BPS.}
For the fixed-width Gaussian-weight specification (BPS-G0) and the
Gaussian-weight specification with estimated widths (BPS-G), the
integration over the latent agent states in
(\ref{eq:supp-conditional-predictive}) is available analytically.
In particular, using the notation of Proposition~\ref{prp:gauss-marginal},
the conditional predictive density for posterior draw $s$ is
\begin{equation}
\begin{split}
p_G^{(s)}(y)
={}&
\left\{
1-\sum_{k=1}^K q_k^{(s)}A_k^{(s)}
\right\}h_0(y)
\\
&+
\sum_{k=1}^K
q_k^{(s)}A_k^{(s)}
\mathcal N
\left(
y\;\middle|\;
a_k^{(s)}-\beta_k^{(s)},
\,
\tau_k^{2(s)}+r_k^{(s)}
\right).
\end{split}
\label{eq:supp-gaussian-density-draw}
\end{equation}
Thus, no simulation of the latent agent states is required when
evaluating the predictive density for these specifications.
The marginal posterior predictive density is obtained by averaging
(\ref{eq:supp-gaussian-density-draw}) over the retained posterior
draws as in (\ref{eq:supp-posterior-predictive-mc}).
For BPS-G0, the window parameters $\sigma_k^2$ entering
$A_k^{(s)}$, $a_k^{(s)}$, and $r_k^{(s)}$ are fixed, whereas they vary
across posterior draws for BPS-G.

\paragraph{Softmax-weight BPS.}
For the linear and quadratic softmax specifications (BPS-S1 and
BPS-S2), the softmax normalization depends jointly on the latent agent
states, and the integral in
(\ref{eq:supp-conditional-predictive}) does not have the same closed
form. We therefore use Monte Carlo integration.
For a fixed posterior draw $\psi^{(s)}$, let
\[
    x^{(1)},\ldots,x^{(M)}
    \overset{\mathrm{iid}}{\sim}
    \prod_{k=1}^K h_k(x_k).
\]
Then
\begin{equation*}
\begin{split}
\widehat p_S^{(s)}(y)
=
\frac{1}{M}
\sum_{m=1}^M
\Bigg[
&
w_0(x^{(m)};\theta^{(s)})h_0(y)
\\
&+
\sum_{k=1}^K
w_k(x^{(m)};\theta^{(s)})
\mathcal N
\left(
y\;\middle|\;
x_k^{(m)}-\beta_k^{(s)},
\tau_k^{2(s)}
\right)
\Bigg].
\end{split}
\end{equation*}
The allocation variable is marginalized analytically through the
mixture probabilities $w_0,\ldots,w_K$ rather than being sampled.
The marginal posterior predictive density is then approximated by
\begin{equation*}
\widehat p_S(y\mid\mathcal H,\mathcal D)
=
\frac{1}{S}
\sum_{s=1}^S
\widehat p_S^{(s)}(y).
\end{equation*}
For the log-score calculation, the same target is evaluated by a
Rao--Blackwellized Monte Carlo average over the retained posterior
draws and independently generated latent agent states. 

\paragraph{Posterior uncertainty in the predictive density.}
For each retained posterior draw
$\psi^{(s)}=(\theta^{(s)},\beta^{(s)},\tau^{2(s)})$,
let
\[
p^{(s)}(y)
=
p(y\mid\mathcal H,\psi^{(s)})
\]
denote the predictive density conditional on that draw.
For BPS-G0 and BPS-G, $p^{(s)}(y)$ is available analytically from
(\ref{eq:supp-gaussian-density-draw}).

For BPS-S1 and BPS-S2, we approximate the latent-state integral by
\[
\widehat p^{(s)}(y)
=
\frac{1}{M}
\sum_{m=1}^{M}
\left[
w_0(x^{(m)};\theta^{(s)})h_0(y)
+
\sum_{k=1}^{K}
w_k(x^{(m)};\theta^{(s)})
\mathcal N
\left(
y\mid x_k^{(m)}-\beta_k^{(s)},\tau_k^{2(s)}
\right)
\right],
\]
where
\[
x^{(m)}
\overset{\mathrm{iid}}{\sim}
\prod_{k=1}^{K}h_k(x_k),
\qquad m=1,\ldots,M.
\]
We set $M=5{,}000$ and use the same latent-state draws
$\{x^{(m)}\}_{m=1}^{M}$ for all posterior draws.

Using $S_{\rm plot}=1{,}000$ posterior draws, the predictive-density
figures report
\[
\overline p(y)
=
\frac{1}{S_{\rm plot}}
\sum_{s=1}^{S_{\rm plot}} p^{(s)}(y)
\]
as the solid curve and
\[
\left[
Q_{0.15}\{p^{(s)}(y)\}_{s=1}^{S_{\rm plot}},
\,
Q_{0.85}\{p^{(s)}(y)\}_{s=1}^{S_{\rm plot}}
\right]
\]
as the 70\% pointwise posterior credible band, where
$p^{(s)}(y)$ is replaced by $\widehat p^{(s)}(y)$ for the softmax
specifications.
The band therefore quantifies posterior uncertainty about the
predictive density, rather than predictive uncertainty for a future
observation.

\paragraph{Weight-function figures.}
For the weight-function figures, we use a common scalar grid
$\mathcal G=\{g_1,\ldots,g_L\}$. For the Gaussian-weight
specifications, the two agent weights are evaluated at
$x_1=x_2=g_\ell$. For the linear and quadratic softmax specifications,
whose weights depend jointly on $(x_1,x_2)$, the displayed curves are
likewise one-dimensional slices along the diagonal
$x_1=x_2=g_\ell$.

For the four BPS specifications, the solid curves are posterior
medians and the shaded bands are the 15th--85th posterior quantile
intervals.The linear pool has constant weights. For GDFC, the figures
display the pointwise-normalized coefficients
$\bar w_j^{\mathrm{GDFC}}(y)$ defined above, rather than the globally
normalized coefficients $\tilde\nu_{js}$ that enter the predictive
density.
For GDFC, the figures display the pointwise-normalized coefficients $\bar w_j^{\mathrm{GDFC}}(y)$ defined above, rather than the globally normalized coefficients $\tilde\nu_{js}$ that enter the predictive density.

\section{Case study: details}
\label{supp:application}

This section provides the sample construction, prior specifications, and additional results for the theatre revenue application in Section~\ref{sec:application}.

\subsection{Data description}
\label{supp:app-data}

\subsubsection*{Data for the baseline}

\citet{pratt1985repetitive} note that the initial observations may reflect learning by the forecaster and focus on the post-summer forecasts as more stable assessments. 
Following this distinction, we reserve Programs~1--24 for constructing the historical baseline and use observations from Program~25 onward for BPS estimation and evaluation.

\subsubsection*{Outliers and missings}

For Harvard Square Theatre, Program~47 is excluded as in the original study because it is a Christmas-day opera and significantly differs from the regular screening programs. 
We also exclude observations with missing realized revenues: Programs~103--106 for Brattle and Programs~89--91 for Harvard Square. 
The resulting post-summer samples contain 78 and 63 observations, respectively. 
Ordering observations by program number, we allocate the first two thirds to training and the remainder to evaluation, as shown in Table~\ref{tab:app-samples}.

\begin{table}[htbp]
\centering
\caption{Training and evaluation samples for the theatre application.
The historical sample, Programs~1--24, is excluded from both sets.}
\label{tab:app-samples}
\begin{tabular}{lllr}
\toprule
Theatre & Sample & Programs & Observations\\
\midrule
Brattle & Training & 25--76 & 52\\
        & Evaluation & 77--102 & 26\\
Harvard Square & Training & 25--67, excluding 47 & 42\\
               & Evaluation & 68--88 & 21\\
\bottomrule
\end{tabular}
\end{table}

\subsubsection*{Transformation of raw data}

As discussed in the main text, we scale the raw forecasts/observations of revenues to account for differences in program duration. For program $i$, let $N_i$ be the number of screening days, $Y_i^{\mathrm{gross}}$ the realized gross revenue, and $Q_{ir}$ the reported forecast at level $r$. The response and inputs are defined as 
\begin{equation*}
    y_i=\log\left(\frac{Y_i^{\mathrm{gross}}}{N_i}\right),
    \qquad
    x_{ir}=\log\left(\frac{Q_{ir}}{N_i}\right),
    \quad r\in\mathcal R.
\end{equation*}

Figure~\ref{fig:app-histograms} shows the distributions of log daily revenue in the training and evaluation samples. 
For Brattle, the evaluation sample has a somewhat lower mean and a narrower observed range. 
For Harvard Square, the centers of the two samples are similar, while the evaluation sample includes a relatively large observation in the right tail. 
The historical observations used to construct the baseline are not included in these histograms.

\begin{figure}[htbp]
\centering
\includegraphics[width=0.85\linewidth]
{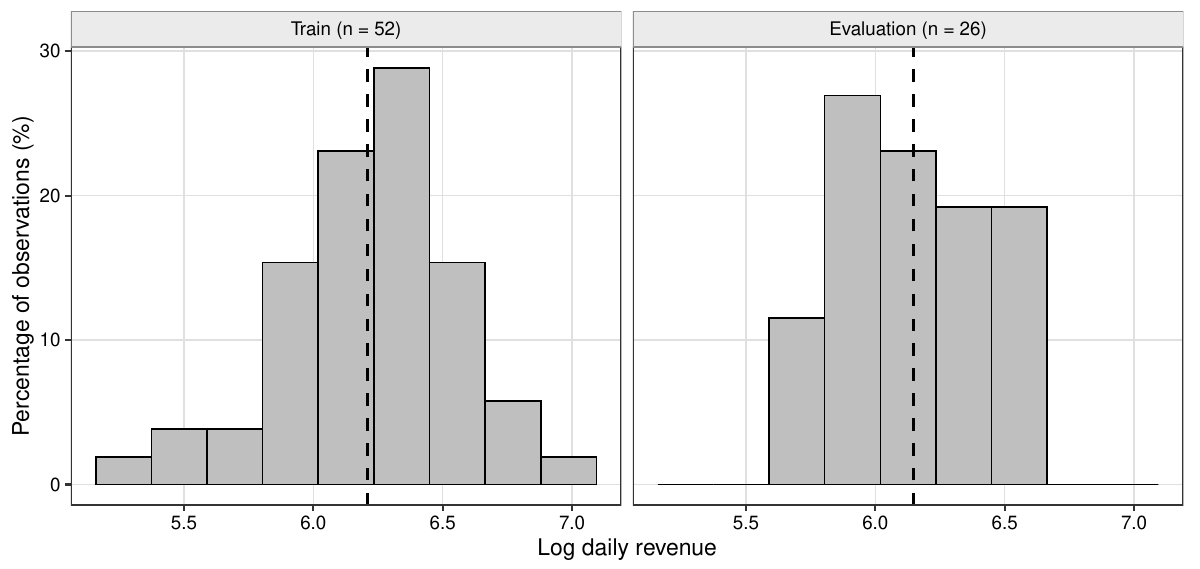}\\[0.6em]
\includegraphics[width=0.85\linewidth]
{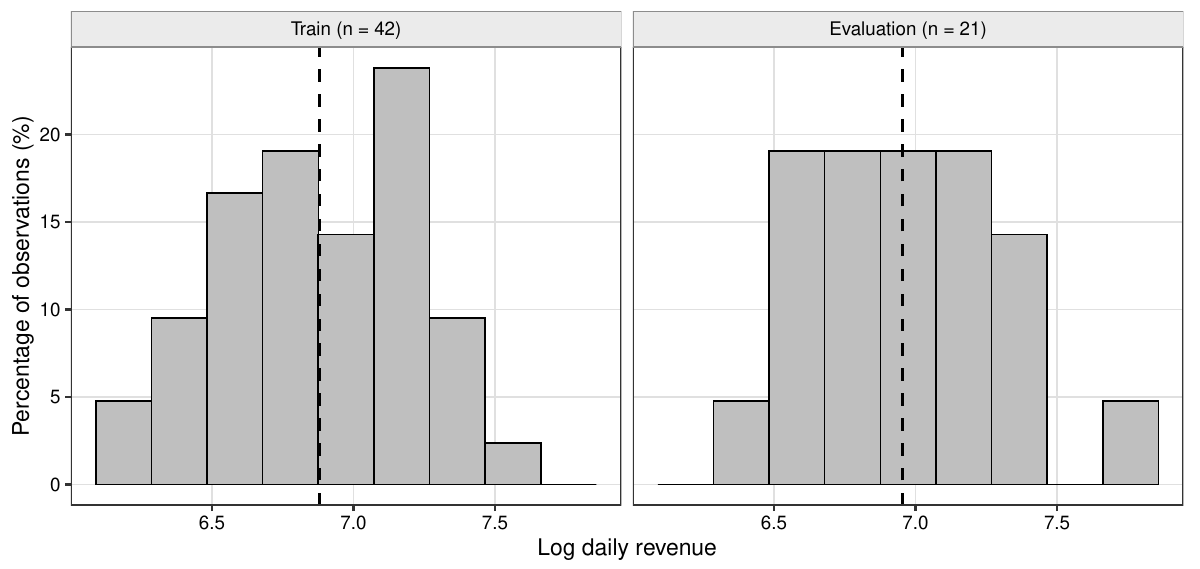}
\caption{Distributions of log daily revenue in the post-summer samples
for Brattle Theatre (top) and Harvard Square Theatre (bottom).
Within each row, the training sample is on the left and the evaluation
sample on the right. Bar heights are percentages of observations
within each sample, and vertical dashed lines mark sample means.}
\label{fig:app-histograms}
\end{figure}

\subsection{Historical baseline and prior specification}
\label{supp:app-priors}

\paragraph{Historical baseline and input standardization.}
For each theatre, let $\bar y_{\mathrm{pre}}$ and $s_{\mathrm{pre}}^2$ denote the sample mean and variance of log daily revenue in Programs~1--24. 
We fix
\[
    h_0(y)=\mathcal N(y\mid m_0,s_0^2),
    \qquad m_0=\bar y_{\mathrm{pre}},
    \quad s_0^2=s_{\mathrm{pre}}^2.
\]
The resulting baselines are
\begin{align*}
    h_0^{\mathrm{Brattle}}(y)
    &=\mathcal N(y\mid6.2557,0.1180),\\
    h_0^{\mathrm{Harvard}}(y)
    &=\mathcal N(y\mid6.8794,0.0570).
\end{align*}
The second argument after the mean is the variance.
These densities provide historical forecast components that do not use the current subjective reports. 
The pre-summer observations do not enter the likelihood for the BPS parameters; they determine the baseline and the reference scale used for prior specification.

Define
\begin{equation}
    u_{ir}=\frac{x_{ir}-m_0}{s_0},
    \qquad d_r=\Phi^{-1}(r),
    \qquad \widetilde u_{ir}=u_{ir}-d_r,
    \quad r\in\mathcal R,
    \label{eq:app-standardization}
\end{equation}
where $\Phi$ is the standard normal distribution function. 
In the order of the probability levels in $\mathcal R$,
\[
    (d_r)_{r\in\mathcal R}=(-2.326,-0.674,0,0.674,2.326).
\]
Thus, $\widetilde u_{ir}=0$ corresponds to the historical baseline's
$r$-quantile. 
Standardization applies to the weight functions only; the response and the Gaussian calibration components in (\ref{eq:deg_ba}) remain on the log daily revenue scale.

\paragraph{Gaussian weights.}
For this application, we use
\begin{equation*}
    w_{ir}=0.85\,q_r
    \exp\left\{-\frac{(u_{ir}-\mu_r)^2}{2\sigma_r^2}\right\},
    \qquad
    w_{i0}=1-\sum_{r\in\mathcal R}w_{ir},
\end{equation*}
where $w_{ir}=w_r(x_i;\theta)$ and $q=(q_r)_{r\in\mathcal R}$.
The factor $0.85$ retains a baseline weight of at least $0.15$ for every input vector. 
The priors are
\[
    q\sim\mathrm{Dir}(2,2,2,2,2),\qquad
    \mu_r\sim\mathcal N(d_r,0.5^2),\qquad
    \sigma_r^2\sim\mathrm{IG}(4,8.656).
\]
Hence,
\[
    E(\sigma_r^2)=\frac{8.656}{3}
    \simeq2.8853\simeq\frac{2}{\log 2}.
\]
At $\mu_r=d_r$ and this reference width, the Gaussian kernel falls to approximately half its maximum at a distance of two historical standard deviations from the reference location.

\paragraph{Softmax weights.}
Set the baseline score to $\eta_{i0}=0$. The linear and quadratic specifications use, respectively,
\begin{align*}
    \eta_{ir}&=A_r+B_r\widetilde u_{ir},\\
    \eta_{ir}&=A_r+B_r\widetilde u_{ir}
                 +C_r\widetilde u_{ir}^{\,2}.
\end{align*}
The weights are
\begin{equation*}
    w_{ir}=\frac{\exp(\eta_{ir})}
    {1+\sum_{s\in\mathcal R}\exp(\eta_{is})},
    \qquad
    w_{i0}=\frac{1}
    {1+\sum_{s\in\mathcal R}\exp(\eta_{is})}.
\end{equation*}
The affine transformation in (\ref{eq:app-standardization}) reparameterizes the linear and quadratic score families without changing their polynomial order.

To specify the intercept prior, let $b_0=0.15$ and
\[
    (\pi_r)_{r\in\mathcal R}
    =(0.0625,0.25,0.375,0.25,0.0625),
    \qquad
    a_r=\log\left\{\frac{(1-b_0)\pi_r}{b_0}\right\}.
\]
These choices give
\[
    (a_r)_{r\in\mathcal R}
    \simeq(-1.038,0.348,0.754,0.348,-1.038).
\]
For both softmax specifications, we use
\[
    A_r\sim\mathcal N(a_r,0.5^2),
    \qquad B_r\sim\mathcal N(0,0.25^2).
\]
At the reference vector $\widetilde u_{ir}=0$ for every $r$ and with $A_r=a_r$, the baseline weight is $b_0=0.15$ and the remaining $0.85$ is allocated according to $\pi_r$.
This is the allocation obtained by fixing the coefficients at their prior means, rather than the prior expectation of the normalized weights. 
Also, $0.15$ is a reference allocation for softmax, whereas it is a lower bound on the Gaussian baseline weight.

For BPS-S2, we additionally specify
\[
    C_r\sim\mathcal N(-0.173,0.08^2).
\]
At $B_r=0$ and $C_r=-0.173$, moving from
$\widetilde u_{ir}=0$ to $\widetilde u_{ir}=\pm2$ multiplies the component-to-baseline odds $w_{ir}/w_{i0}$ by $\exp\{4(-0.173)\}\simeq0.5$.
Thus, the prior favors local use of each input near its historical reference quantile. 
This interpretation concerns the odds evaluated at the prior mean coefficients; the normalized weights also depend on the scores of the competing inputs. The normal prior encourages, but does not impose, negative curvature.

\paragraph{Calibration parameters.}
All three BPS specifications use
\[
    \beta_r\sim\mathcal N(0,0.1^2),
    \qquad \tau_r^2\sim\mathrm{IG}(10,0.45),
    \quad r\in\mathcal R.
\]
These priors favor small location adjustments and have $E(\tau_r^2)=0.05$ on the log daily revenue scale.

\subsection{Computation, quantile matching, and predictive evaluation}
\label{supp:app-evaluation}

Let $\mathcal T$ and $\mathcal E$ denote the training and evaluation index sets for a given theatre, and let $\mathcal D_{\mathcal T}=\{(y_i,x_i):i\in\mathcal T\}$.
The synthesis parameters are common across programs within each theatre and are estimated using
\[
    \pi(\psi\mid\mathcal D_{\mathcal T})
    \propto\pi(\psi)\prod_{i\in\mathcal T}p_i(y_i\mid\psi),
\]
where $p_i$ is given in (\ref{eq:deg_ba}). 
For each theatre and BPS specification, we run MCMC for $100{,}000$ iterations and discard the first $25{,}000$ as burn-in. 
Since the inputs are point masses, the latent agent states are fixed at the reported values.

For $i\in\mathcal E$, the posterior predictive density is
\begin{equation*}
    \overline p_i(y)
    =\int p_i(y\mid\psi)
    \pi(\psi\mid\mathcal D_{\mathcal T})\,d\psi
    \simeq\frac{1}{M}\sum_{m=1}^{M}p_i(y\mid\psi^{(m)}),
\end{equation*}
where $\psi^{(1)},\ldots,\psi^{(M)}$ are retained posterior draws used for prediction. 

For the QM benchmark, we fit a normal distribution to each program's reported quantiles by solving
\begin{equation*}
    (\widehat\mu_i,\widehat\sigma_i)
    =\arg\min_{\mu,\,\sigma>0}
    \sum_{r\in\mathcal R}
    \left[x_{ir}-\{\mu+\sigma\Phi^{-1}(r)\}\right]^2,
\end{equation*}
and set
$p_i^{\mathrm{QM}}(y)=\mathcal N(y\mid\widehat\mu_i,\widehat\sigma_i^2)$.
The Low and High reports are assigned the nominal probability levels $0.01$ and $0.99$ for this fit, as in the BPS prior construction.

For a predictive density $f_i$ and its distribution function $F_i$, the reported scores are
\begin{align*}
    \mathrm{LogS}
    &=\frac{1}{|\mathcal E|}
      \sum_{i\in\mathcal E}\log f_i(y_i),\\
    \mathrm{CRPS}
    &=\frac{1}{|\mathcal E|}\sum_{i\in\mathcal E}
      \int_{-\infty}^{\infty}
      \left\{F_i(t)-\mathbf 1(y_i\leq t)\right\}^{2}\,dt.
\end{align*}
We use $f_i=\overline p_i$ for BPS and $f_i=p_i^{\mathrm{QM}}$ for QM. 
Thus, the BPS log score is the log of the posterior predictive density, with parameter uncertainty integrated out before taking the logarithm. All scores refer to log daily revenue.

\subsection{Additional posterior and predictive summaries}
\label{supp:app-additional}

For BPS-S2, the posterior medians of the curvature coefficients $C_r$ range approximately from $-0.182$ to $-0.164$ for Brattle and from $-0.182$ to $-0.163$ for Harvard Square. 
All five component scores therefore have negative curvature at these posterior
summaries. 
These values are close to the prior mean of $-0.173$, and the retained local structure should be interpreted in light of that prior specification. 
For BPS-S1, the posterior slope summaries are generally near zero for both theatres, consistent with the mild input dependence in Figure~\ref{fig:app-weights}.

The weight plots use a common scalar grid and evaluate each posterior draw at $x_i=(g,g,g,g,g)^{\top}$. 
The transformations in (\ref{eq:app-standardization}) are applied separately for each input at every grid point. This diagonal is a diagnostic slice of the weight functions; the program-specific predictions instead use the observed vectors of ordered quantile reports.
Curves are pointwise posterior medians and bands span the 15th to 85th posterior quantiles. 
The weights sum to one at every posterior draw, although their separately computed posterior medians need not do so.

In the predictive-density figures, the BPS curves average the conditional densities $p_i(y\mid\psi)$ over posterior draws.
At each $y$, the bands span the 15th to 85th posterior quantiles of these conditional density values. They are 70\% pointwise posterior credible bands for the density, rather than prediction intervals for a future revenue observation.

Figure~\ref{fig:app-density-harvard} shows Harvard Square Program~68, the first observation in its evaluation sample.
All three BPS specifications give higher log densities and smaller CRPS values than QM. 
For this program, BPS-G attains both the highest log density, $-0.0152$, and the smallest CRPS, $0.1287$.
Thus, the method ranking for an individual program can differ from the ranking of average scores in Table~\ref{tab:app-performance}.

\begin{figure}[htbp]
\centering
\includegraphics[width=0.88\linewidth]
{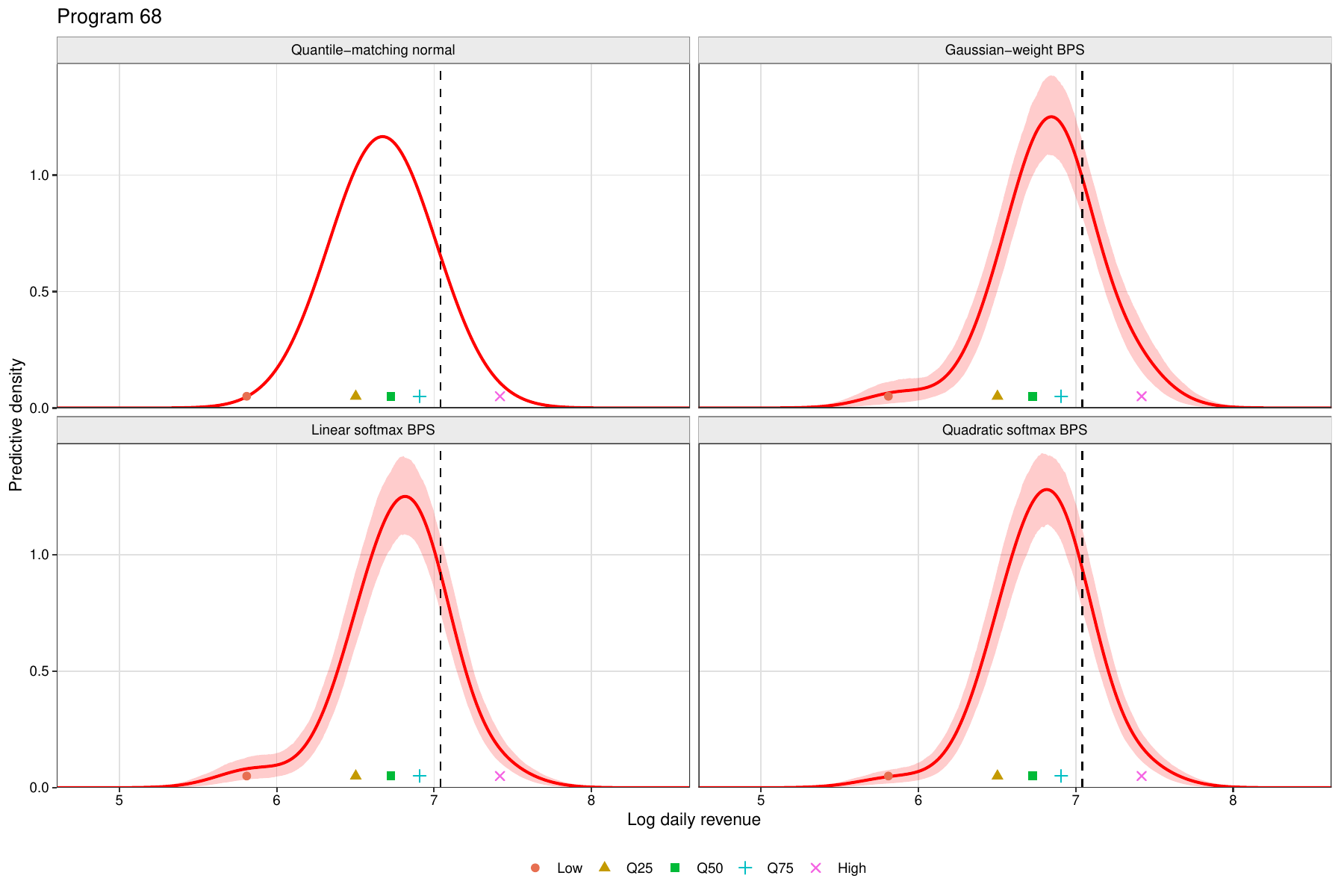}
\caption{Predictive densities for Harvard Square Theatre, Program~68.
For BPS, red curves are posterior means of the conditional predictive
densities and shaded regions are the 15th--85th pointwise posterior
quantile bands. The QM panel shows the fitted normal density.
Markers indicate the five reported inputs, and the vertical dashed
line marks the realized log daily revenue.}
\label{fig:app-density-harvard}
\end{figure}

\end{document}